\documentclass{optica-article}

\journal{opticajournal} 
\articletype{Research Article}
\usepackage{graphicx,amsmath}

\usepackage{color}
\usepackage{lineno}
\usepackage{siunitx}

\DeclareMathOperator*{\argmin}{arg\,min}
\DeclareMathOperator{\Rop}{\mathcal{R}}

\DeclareMathOperator{\Pop}{\mathcal{D}}

\DeclareMathOperator{\Sop}{\mathcal{S}}
\DeclareMathOperator{\Mop}{\mathcal{M}}
\newtheorem{theorem}{Theorem}[section]
\newtheorem{proposition}[theorem]{Proposition}

\newenvironment{proof}{\noindent\textit{Proof:}}{\hfill$\square$}

\def\mod#1{{{#1}}} 

\begin{document}

\title{X-ray nano-holotomography reconstruction with simultaneous probe retrieval}
\author{Viktor Nikitin,\authormark{1,*} Marcus Carlsson,\authormark{2} Do\u{g}a G\"ursoy,\authormark{1} Rajmund Mokso,\authormark{3} and Peter Cloetens\authormark{4}}

\address{\authormark{1} Advanced Photon Source, Argonne National Laboratory, 9700 S Cass Ave, Lemont, IL 60439, USA\\
\authormark{2} Centre for Mathematical Sciences, Lund University, Sölvegatan 18, 223 62 Lund, Sweden\\
\authormark{3} Department of Physics, Danish Technical University, Fysikvej, 310, 2800 Kgs. Lyngby, Denmark\\
\authormark{4} ESRF-The European Synchrotron
71, Avenue des Martyrs, 38043 Grenoble, France\\
}

\email{\authormark{*}vnikitin@anl.gov} 

\begin{abstract*} 
In conventional tomographic reconstruction, the pre-processing step includes flat-field correction, where each sample projection on the detector is divided by a reference image taken without the sample. When using coherent X-rays as probe, this approach overlooks the phase component of the illumination field (probe), leading to artifacts in phase-retrieved projection images, which are then propagated to the reconstructed 3D sample representation. The problem intensifies in nano-holotomography with focusing optics, that due to various imperfections create high-frequency components in the probe function. Here, we present a new iterative reconstruction scheme for holotomography, simultaneously retrieving the complex-valued probe function. Implemented on GPUs, this algorithm results in 3D reconstruction resolving twice thinner layers in a 3D ALD standard sample measured using nano-holotomography.
\end{abstract*}

\section{Introduction}
Holotomography is an advanced imaging technique that extends conventional X-ray tomography by exploiting the phase shift that X-rays undergo when passing through a sample~\cite{cloetens1999holotomography}. This technique allows for visualizing phase components of the internal structure of materials with high resolution and contrast. It is particularly valuable for studying samples that have low absorption contrast, such as biological tissues~\cite{langer2012x,kuan2020dense,andersson2020axon,robisch2020nanoscale} or battery materials \cite{nguyen20213d,li2022dynamics}. 

The fundamental principle of holotomography lies in measuring phase shifts experienced by X-rays as they traverse different materials, converted into phase-contrast images using various algorithms. As synchrotron sources achieve higher coherence, the technique faces challenges in reaching nano-scale resolution~\cite{martinez2016id16b,da2017high}. Nano-holotomography employs the Projection X-ray Microscope (PXM) instrument, utilizing focused X-ray beams to produce magnified projections of a sample. A schematic of the PXM is shown in Figure~\ref{fig:setup}. 
\mod{To recover accurate 3D phase-contrast sample structures, tomography data acquisition with the sample rotation is typically performed at several distances from the focal spot, marked in the figure as $z_j,\,j=0,\dots,N_z-1$. The distance between the focal spot and detector is denoted as $Z$, which gives geometric magnifications $m_j=Z/z_j,$ on the detector due to the cone-beam geometry.
The feasibility of this multi-distance approach for quantitative retrieval of the object phase was first demonstrated by Cloetens et al.~\cite{cloetens1999holotomography}. In practice, four distances are typically sufficient for accurate phase reconstruction. Within the cone-beam geometry employed in the nano-holotomography regime, the initial distance is selected to achieve the required resolution. Subsequent distances are calculated as fixed ratios of the first distance, forming the basis for multi-distance holotomography measurements at beamlines.}

With the theoretical foundations firmly established, PXM instruments implementing multi-distance holotomography are becoming increasingly prevalent at synchrotrons worldwide. Notable examples include beamlines at the ESRF (ID16A~\cite{da2017high} and ID16B~\cite{martinez2016id16b}), Max IV (NanoMAX~\cite{kalbfleisch2022x}), and DESY (P10~\cite{hagemann2017probe}). Additionally, a new PXM is currently under construction at the APS (32-ID~\cite{bean2021new}).  

Kirkpatrick-Baez (KB) mirrors, normally used for high-resolution imaging, use reflective surfaces to focus X-rays in two perpendicular directions. Multilayer coatings significantly enhance the performance of KB mirrors, offering increased reflectivity for higher intensity X-ray beams, enhanced focusing precision for higher resolution, and optimization across a broader range of wavelengths. 
While the multilayer coating on KBs offers many benefits, it also creates high-frequency components into the illumination field (probe), thereby complicating the phase retrieval process. The right panel of Figure~\ref{fig:setup} shows examples of a reference image and acquired projection data for different distances from the sample, with each projection containing vertical and horizontal lines caused by the multilayer coating.

\begin{figure}[t]
\centering\includegraphics[width=1\textwidth]{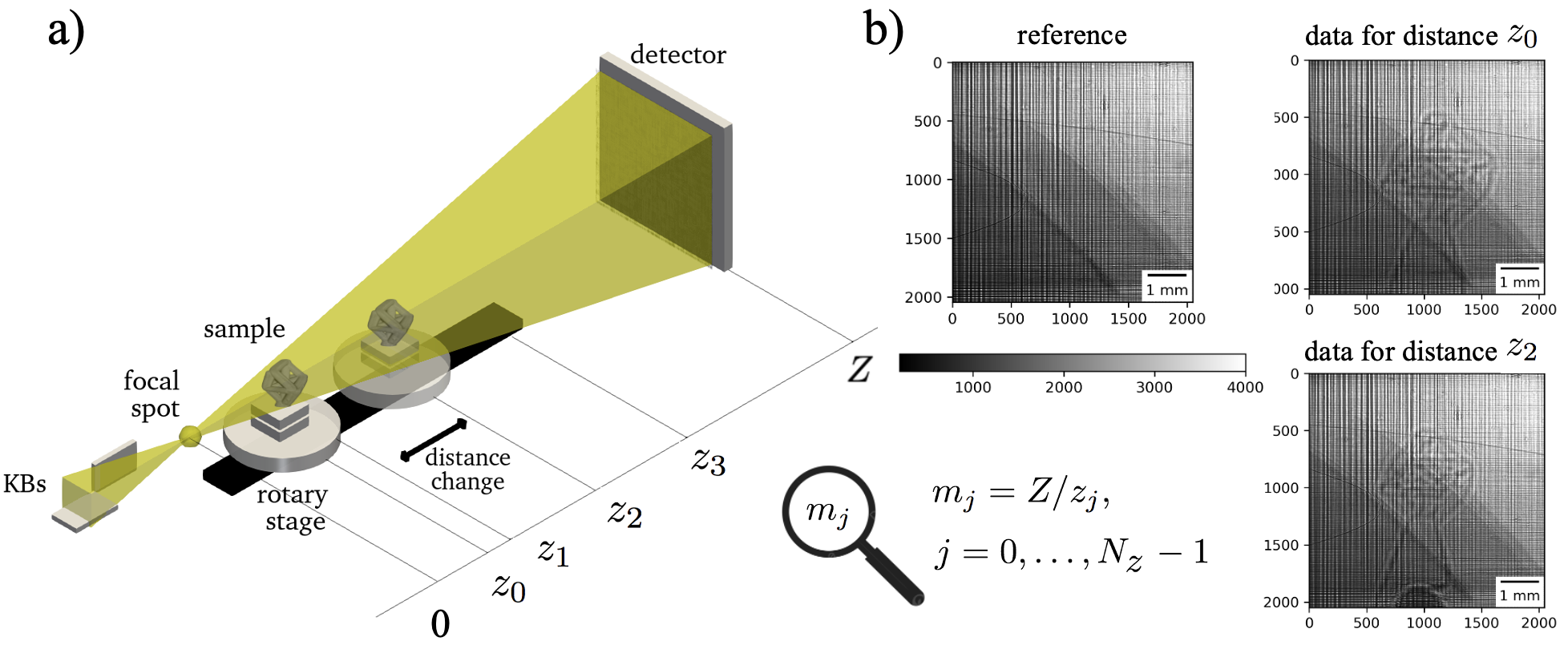}
\caption{Holotomography data acquisition with the PXM: a) a scheme of data collection with 4 distances from the focal spot $\left(z_j, j=0,\dots,3\right)$, with corresponding geometric magnifications $m_j$ of images at the detector position $Z$,  b) an example of collected data at ID16A beamline of ESRF: reference image on the detector when the sample is out of the field of view (left) and data on the detector for the first and fourth distances between the focal spot and sample planes (right).}
\label{fig:setup}
\end{figure}

The conventional phase-retrieval procedure for a single projection involves division by the reference image, followed by reconstruction using phase-retrieval methods such as the transport-of-intensity (TIE) equation~\cite{gureyev1996phase,barty1998quantitative} or contrast transfer function (CTF)~\cite{cloetens1999holotomography,zabler2005optimization}, supplemented by iterative data fitting employing the Fresnel propagation operator~\cite{paganin2006coherent}. These approaches, along with their extensions, are comprehensively reviewed and compared in the literature, see e.g.~\cite{langer2008quantitative}. 
However, the initial division by the reference image fails to account for the phase component of the probe, aggregated by the oscillatory structure caused by the multilayer coating of the KB mirrors. These factors contribute to imaging artifacts and \mod{are} the bottleneck in achieving high-resolution nano-scale imaging.

A potential hardware solution to homogenize  the illumination is to use X-ray waveguide optics~\cite{salditt2008high, salditt2015compound}. Waveguides are structures designed to guide X-rays with minimal dispersion, maintaining coherence and enhancing the imaging process. By properly adjusting the waveguide-to-sample distance, objects can be imaged at a single distance in a full-field configuration, eliminating the need for scanning. Robust and rapidly converging iterative reconstruction schemes can then be applied to invert the holographic data. \mod{However, utilizing X-ray waveguides presents challenges such as complex fabrication ~\cite{kruger2012sub} and a significant reduction in beam intensity at the waveguide exit (see Figure 3 in~\cite{salditt2015compound}).}

Hagemann et al. \cite{hagemann2014reconstruction} were likely the first to recognize problems related to flat-field division, proposing methods similar to coherent diffractive imaging and ptychographic reconstruction algorithms as solutions. Their work pioneered treating the probe retrieval problem as part of phase retrieval in holography, using iterative schemes like the extended Ptychographic Iterative Engine (ePIE) \cite{maiden2009improved}, relaxed averaged alternating reflections (RAAR) \cite{luke2004relaxed}, and solvers like PtyPy \cite{enders2016computational}. Their work focused on 2D holography, where limited or under-sampled data created challenges for reconstructing two complex-valued images, affecting the stability of the reconstruction process. To address this, they proposed adding constraints to stabilize the solution \cite{hagemann2017probe}. In a later work \cite{dora2024artifact}, they deviated from support constraints and derived an optimization based approach similar to the one proposed in this paper, albeit for parallel beam geometry and without taking into consideration probe retrieval.

In this work, we adapt the mathematical framework to treat cone-beam geometry as well as probe retrieval and extend the problem to 3D, leveraging more data from tomography. We introduce a new mathematical formulation and an iterative algorithm optimized for modern Graphical Processing Units (GPUs). This approach allows for efficient recovery of 3D objects and widespread application to large-scale synchrotron datasets, eliminating the need for additional support constraints to maintain numerical stability. These advancements are timely, given the spread of coherent imaging instrumentation worldwide.

The remainder of the paper is structured as follows. Section 2 presents the formulation of the holotomography reconstruction problem with probe retrieval, elucidating the solution through an iterative scheme. In Section 3, we present validation results using synthetic data, including convergence analysis and assessment of GPU acceleration. Section 4 showcases results obtained from experimental data acquired with the PXM instrument at the ID16A beamline of ESRF, and includes a comparative analysis with results obtained using conventional approaches. Section 5 offers conclusions and outlines future prospects.

\section{Holotomography problem with probe retrieval}

We will begin with the conventional formulation of the 2D phase retrieval problem with parallel beam. In this formulation, the acquired data is divided by the reference image. For the object transmittance function $\psi$ the formulation reads as
\begin{equation}\label{Eq:fwd_model}
  \left|\Pop_{Z-z_j} (\psi)\right|^2 = d_{j}/d^r, 
\end{equation}
where $d_{j}$ denotes data for sample-detector distances $Z-z_j$, $(j=0\dots N_z-1),$ and $\Pop_{Z-z_j}$ is the corresponding Fresnel transform defined via convolution with the Fresnel kernel (see equation \eqref{defPz} in the Appendix), 
while $d^r$ is the reference data, i.e. the squared absolute value of the probe function on the detector when the sample is out of field of view. The division of two images (i.e.~matrices) is to be interpreted pixelwise, and analogously the above modulus is also applied elementwise to every pixel. However, as pointed out in \cite{hagemann2014reconstruction,hagemann2017x}, the above formula is only a crude approximation of the physical process and introduces severe distortion effects.

More precisely, assume that at position $z_0$ we have probe function $q$. Then the probe image on the detector can be described as 
\begin{equation}
    \left|\Pop_{Z-z_0}(q)\right|^2 = d^r.   
\end{equation}
With the object in position $z_0$ as well, we then measure $d_0=\left|\Pop_{Z-z_0} (q\cdot \psi)\right|^2$ which clearly does not equal $\left|\Pop_{Z-z_0} (q)\cdot \Pop_{Z-z_0}(\psi)\right|^2$, which would need to be true in order for \eqref{Eq:fwd_model} to hold, (see\cite{homann2015validity} for an in-depth mathematical study of these issues).

The data modeling process for general $j$ should include two wave propagations: (1) propagation of the probe function to the sample plane $z_j$ and (2), propagation of the wave formed after interaction of the probe and the sample to the detector. It follows that a new holography model should be written as
\begin{equation}\label{Eq:fwd_model2}
  \left|\Pop_{Z-z_j}\big(\Pop_{z_j-z_0}(q) \cdot \psi\big)\right|^2 = d_{j}. 
\end{equation}
It is easy to see that models \eqref{Eq:fwd_model} and \eqref{Eq:fwd_model2} are not equivalent since the first model does not take into account the phase component of the probe. In Section 3 we will demonstrate this difference using simulated data and show how it affects the final reconstruction.

The new formulation \eqref{Eq:fwd_model2} introduces an additional unknown to the holography problem: the probe function $q$. A viable approach to recover both the sample and the probe is to use the near-field ptychography (NFP) method~\cite{stockmar2013near}, which involves supplementary measurements of a high-contrast, detailed sample, such as a Siemens star. However, we argue that such an approach is suboptimal since the probe can be unstable and may vary significantly between the NFP measurements and the multi-distance holography measurements of the actual sample. \mod{Alternatively, the NFP method could be applied directly to actual samples. However, if the sample lacks sufficient features, the number of scanning $x,y$ positions must be increased to ensure adequate data redundancy and improve probe retrieval accuracy. Moreover, sharp features are typically uncommon in projections through real 3D samples measured during experiments. The resulting increase in scanning positions for NFP leads to additional radiation dose to the sample.} 

We will not investigate further how to best solve \eqref{Eq:fwd_model2}, since the main topic of this article is holotomography in which the above issues disappear, since
solving the holography problem for several projections, $\psi_k, k=0,\dots,N_\theta-1$ gives additional constraints to variable $q$. Assuming that we have a parallel beam and that the probe function is not changed during acquisition of the $N_\theta$ projections, the data acquisition process can be modeled as follows,

\begin{equation}\label{Eq:fwd_model3}
  \left|\Pop_{Z-z_j}\big(\Pop_{z_j-z_0}(q) \cdot \psi_k\big)\right|^2 = d_{k,j}, 
\end{equation}
where $q$ and $\psi_k, k=0,\dots, N_\theta-1$, are unknowns. The $L_2$-norm minimization,
\begin{equation}\label{Eq:min0}
\argmin_{q,\psi_k }\sum_{k,j}\left\|\left|\Pop_{Z-z_j}\big(\Pop_{z_j-z_0}(q) \cdot \psi_k\big)\right|^2 - d_{k,j}\right\|_2^2,
\end{equation}
finalizes our problem formulation for the parallel beam geometry. 

The problem formulation for the cone-beam geometry in nano-resolution imaging leverages the Fresnel scaling theorem~\cite{paganin2006coherent}, which transforms the problem into an equivalent one in parallel beam geometry. With the conventional reconstruction approach involving division by the reference image, the formulation reads as 
\begin{equation}\label{Eq:fwd_model_cone}
  \left| \Mop_{m_j}\big( \Pop_{\zeta_j}  (\psi )\big)\right|^2 =  {d_{j}}/{d^r},  
\end{equation}
where $\Mop_{m_j} (f)(x,y) = f\left({x}/{m_j},{y}/{m_j}\right)$ with $m_j = {Z}/{z_j}$, and adjusted distances are calculated as $\zeta_j={(Z-z_j)}/{m_j}$. In practice, the magnification factors are moved to the right side of the formula,
\begin{equation}\label{Eq:fwd_model_cone2}
  \left| \Pop_{\zeta_j} (\psi) \right|^2 =  \Mop_{1/m_j}({d_{j}}/{d^r})
\end{equation}
which, apart from the data scaling, becomes identical to the parallel beam formulation~\eqref{Eq:fwd_model}.

For the problem with the probe retrieval, the setup is more intricate and we can not simply move the magnification factor to the data. This issue stems from the fact that we are applying the Fresnel scaling theorem twice, once when we propagate the probe to the sample position and then again when propagating the product to the detector. However, we prove in Appendix A that this chain of operations can be evaluated in a way that the rescaling, necessary for the various distances, only affects the sample and not the probe. Numerically, this leads to much more stable reconstruction since the probe typically has large oscillations on the pixel scale and hence interpolating the probe is an unstable operation. We provide the details in Appendix A and here simply present our working formula for the forward model: 
\begin{equation}\label{Eq:fwd_model_cone3}
  \left|\frac{z_0}{Z}\Mop_{m_0}\big(\Pop_{\zeta_j/\tilde m_j^2}\big(\Pop_{\omega_j}(q)\cdot \Mop_{1/\tilde m_j}(\psi) \big)\big)\right|^2 = d_{j},
\end{equation}
where $\tilde m_j={z_j}/{z_0}$ and $\omega_j=(z_j-z_0)/\tilde m_j$ and, as previously, $\zeta_j={z_j(Z-z_j)}/{Z}$ and $m_0={Z}/{z_{0}}$. 

{Note that the different rescalings $\Mop_{1/\tilde m_j}$ in this setup is only affecting $\psi$, rather than the measured data corrected by the reference image in the traditional setup without probe retrieval~\eqref{Eq:fwd_model_cone2}. 
The formula's consistency can be verified by setting 
$q=1$ and, separately, $\psi=1$, both of which lead to the conventional form for cone-beam propagation.} One more constraint to the problem can be added using the reference image data, namely $\left|\frac{z_0}{ Z}\Mop_{m_0}\big(\Pop_{\zeta_0}(q)\big)\right|^2 = d^r$. Finally note that $$\left|\frac{z_0}{ Z}\Mop_{m_0}\big(\Pop_{\zeta_0}(q)\big)\right|^2 = \frac{1}{ m_0^2}\Mop_{m_0}\big(|\Pop_{\zeta_0}(q)|^2\big)$$ so that the above can be simplified to $|\Pop_{\zeta_0}(q)|^2 = \tilde d^r$ where $\tilde d^r=m_0^2\Mop_{1/m_0}(d^r)$.

Experimental data requires alignment for both different sample planes and tomography projections. Also, the probe is often shifting during the data acquisition.  Corresponding shifts for alignment are typically found before reconstruction as a pre-processing step. In the case of conventional reconstruction involving division by the reference image, all measured data can be shifted back according to recorded shifts prior to reconstruction. For the proposed reconstruction, the shift operator can be added to the problem formulation. Generalizing everything for the case with several projections yields the following $L_2$-norm functional for holography with probe retrieval, which needs to be minimized given that the $\psi_k$'s are related through additional constraints given the particular experimental framework:

\begin{align}    
\label{Eq:min2}&F_1(\psi_0,\dots,\psi_{N_\theta - 1},q) =\\ \nonumber&\sum_{k,j}\left\|\left|\Pop_{\zeta_j/\tilde m_j^2}\left(\Pop_{\omega_j}\big(\Sop_{s_{k,j}}(q)\big)\cdot\Mop_{1/\tilde m_j} \big(\Sop_{t_{k,j}}(\psi_k)\big)\right)\right|^2 - \tilde d_{k,j}\right\|_2^2+ \left\|\left|\Pop_{\zeta_0}(q)\right|^2 - \tilde d^r\right\|^2_2,
\end{align}
where $\tilde d_{j,k}$ stands for $m_0^2 \Mop_{1/m_0} d_{j,k}$ and the operator $\Sop_s$, with $s\in\mathbb{R}^2$, is a shift operator for moving functions in $x$ and $y$ coordinates, i.e. orthogonal to the beam. The shifts $t_{k,j}$ are used to move the sample whereas $s_{k,j}$ account for shifts in the probe. We remark that the shift operators commute with all other involved operators, so it is possible to remove the shifts from probe $q$ and instead have one shift that affects $\psi_k$ and another that affects the data $\tilde d_{k,j}$. The benefit of this approach in practice is that it removes the necessity to interpolate the probe (which is highly oscillatory). However, we do not pursue this further here. 
 An alternative minimization functional to \eqref{Eq:min2} is 
 \begin{align}    
 \label{Eq:min1}&F_2(\psi_0,\dots,\psi_{N_\theta - 1},q) =\\ \nonumber&\sum_{k,j}\left\|\left|\Pop_{\zeta_j/\tilde m_j^2}\left(\Pop_{\omega_j}\big(\Sop_{s_{k,j}}(q)\big)\cdot \Mop_{1/\tilde m_j} \big(\Sop_{t_{k,j}}(\psi_k)\big)\right)\right| - \sqrt{\tilde d_{k,j}}\right\|_2^2+\left\|\left|\Pop_{\zeta_0}(q)\right| - \sqrt{\tilde d^r}\right\|^2_2 
 \end{align}
which, based on our tests, is more favorable in terms of the convergence {speed of the Conjugate Gradients solver used in this work.} 
A somewhat lengthy calculation yields that
\begin{equation}\label{Eq:gradpsi}
        \nabla_{\psi_k} F_2(\psi_0,\dots,\psi_{N_\theta - 1},q) = 2\sum_j L_{k,j}^*\left(L_{k,j}(\psi_{k})-\sqrt{\tilde d_{k,j}}\cdot \frac{ L_{k,j}(\psi_{k})}{\left|L_{k,j}(\psi_{k})\right|}\right)
\end{equation}
where $(\cdot)^*$ denotes the Hermitian adjoint operator and $$L_{k,j}(\psi)=\Pop_{\zeta_j/\tilde m_j^2}\left(\Pop_{\omega_j}\big(\Sop_{s_{k,j}}(q)\big)\cdot\Mop_{1/\tilde m_j} \big(\Sop_{t_{k,j}}(\psi)\big)\right).$$ 
After another tedious calculations, it also follows that $$L_{k,j}^*(d)=\Sop_{-t_{k,j}}\left(\Mop_{1/\tilde m_j}^*\left(\overline{(\Pop_{\omega_j}\big(\Sop_{s_{k,j}}(q)\big)}\cdot\Pop_{-\zeta_j/\tilde m_j^2}(d)\right)\right).$$ 
Similarly, the steepest ascent direction for the probe function $q$ is given as
\begin{align*}
            &\nabla_q F_2(\psi_0,\dots,\psi_{N_\theta - 1},q) = \\&2\left(\sum_{k,j} Q_{k,j}^*\left(Q_{k,j}(q)-\sqrt{\tilde d_{k,j}}\cdot\frac{ Q_{k,j}(q)}{\left|Q_{k,j}(q)\right|}\right)
        +\Pop_{\zeta_0}^*\left(\Pop_{\zeta_0}(q)-\sqrt{\tilde d^r}\cdot\frac{ \Pop_{\zeta_0} (q)}{\left|\Pop_{\zeta_0} (q)\right|}\right)\right),
\end{align*}
where  $$Q_{k,j}(q)=\Pop_{\zeta_j/\tilde m_j^2}\left((\Pop_{\omega_j}\big(\Sop_{s_{k,j}}(q)\big)\cdot\Mop_{1/\tilde m_j} \big(\Sop_{t_{k,j}}(\psi_k)\big)\right) .$$ 

Armed with the steepest ascent directions, we can now construct iterative schemes for solving \eqref{Eq:min1} by using a variety of methods. Here we employ a conjugate-gradient (CG) method for its faster convergence rate~\cite{Dai:00}. CG iterations are given as 
\begin{equation}
    \begin{aligned}
        \psi_k^{(m+1)} = \psi_k^{(m)}+ \gamma_k^{(m)} \eta_k^{(m)},
    \end{aligned}
\end{equation}
where $\gamma_k^{(m)}$ is a step length computed by a line-search procedure and $\eta_k^{(m)}$ is the search direction.
Different recursive formulas for $\eta^{(m)}$ have been proposed, including the Fletcher-Reeves~\cite{Fletcher:64}, the Polak-Rib\'ere-\!Polyak~\cite{Polak:69,Polyak:69}, and the Dai-Yuan~\cite{DaiYuan:99} formulas. According to our tests, the Dai-Yuan formula demonstrates a significantly faster convergence rate for this phase retrieval problem. The Dai-Yuan formula is given by 
\begin{equation}\label{Eq:DaiYuan2}
\begin{aligned}
  \eta_k^{(m+1)}= -\nabla_{\psi_k} F(\psi^{(m+1)})+\frac{\|\nabla_{\psi_k} F(\psi_k^{(m+1)})\|_2^2}{\left\langle \left(\nabla_{\psi_k} F(\psi_k^{(m+1)})-\nabla_{\psi_k} F(\psi_k^{(m)})\right),\eta_k^{(m)}\right\rangle}\eta_k^{(m)}
  \end{aligned}
\end{equation}
with $\eta_k^{(0)}=-\nabla_{\psi_k} F(\psi_0)$. Each CG iteration step involves updates of the transmittance function of each angle, $\psi_k, k=0,\dots,N_\theta-1$, as well as the probe function $q$. The line search procedure guaranties that $F\left(\psi_0^{(m+1)},\dots, \psi_{N_\theta-1}^{(m+1)},q^{(m+1)}\right)\le F\left(\psi_0^{(m)},\dots, \psi_{N_\theta-1}^{(m)},q^{(m)}\right)$ on each iteration.

The initial guess for iterative schemes in holography problems significantly impacts the final reconstruction as it influences the optimization process. A poor initial guess can slow down convergence, primarily due to the slow reconstruction of low spatial frequencies, and increase the likelihood of the algorithm converging to a local minimum. This issue has been studied in~\cite{dora2024artifact}, where the authors explored methods with and without additional constraints, regularization, and Nesterov accelerated gradient to counteract the slow reconstruction of low spatial frequencies. In this work, we also rely on the initial guess for reconstruction. As an initial estimate for the holographic reconstruction of the sample transmittance function ($\psi_k, k=0,\dots,N_\theta-1$) we utilized the reconstruction generated by the Transport of Intensity Equation (TIE) method, commonly referred to as the Multi-Paganin method in the context of synchrotron beamlines. This approach was applied to the data following normalization by the reference image. It is worth noting that the Multi-Paganin method typically struggles to recover high-frequency components. Consequently, the presence of horizontal and vertical line artifacts in the data does not significantly impact the reconstruction quality. For the initial guess of the probe function $q$, we used the reference image propagated back to the sample plane 0. Although for processing experimental data, an initial probe approximation can be retrieved from previous measurements.

After reconstruction of all projections $\psi_k$, we can recover a 3D representation of the sample refractive index $u=\delta+i\beta$ by solving a tomography problem. 
We express the minimization of the tomography problem as
\begin{equation}\label{Eq:tomo}
  F(u)= \sum_k\left\|\Rop_k (u)-\frac{\nu}{2\pi i}\log(\psi_k)\right\|^2_2,
\end{equation}
where $\Rop_k$ is the Radon transform with respect to angle $\theta_k$, and $\nu$ is the wavelength.
Similar to the phase retrieval problem, the steepest ascent direction $\nabla_u F(u)$ can be obtained as follows,
\begin{equation}
  \nabla_u F(u) = 2\sum_k\Rop^*_k\left(\Rop_k (u) - \frac{\nu}{2\pi i}\log(\psi_k)\right).
  \label{Eq:gradtomo}
\end{equation}
The tomography problem can thus be also solved by considering the rapid convergence of CG methods that utilize the direction $\nabla_u F(u)$. For a consistent approach with the phase-retrieval problem, we employ the same Dai-Yuan formula given in equation~\eqref{Eq:DaiYuan2}.

Here, we assume that \mod{the phase shift of} $\psi_k$ takes the angle values in 0 to $2\pi$ range. Note that for larger objects and lower X-ray energies, phase unwrapping may be necessary. Apart from general unwrapping procedures~\cite{goldstein1988satellite,ghiglia1998two}, it is also possible to apply correction during reconstruction. One approach involves using a linear model, such as the TIE or CTF, to obtain an initial phase estimate without wrapping issues. During subsequent iterative optimization, it is assumed that the phase change remains small (less than $2\pi$) relative to the linear model. Consequently, phase unwrapping can be performed using the initial phase estimate after all iterations, although it is also possible to unwrap after each iteration.
Another approach leverages the integration with tomography, where explicit phase unwrapping becomes unnecessary. Reconstruction can be achieved from the derivative of the phase using a modified filter in the filtered back projection (FBP) algorithm, specifically a Hilbert filter. The derivative of the phase can be accurately approximated from the complex object alone. This method is detailed in the paper by researchers from cSAXS beamline (equations (12) and (13) in ~\cite{guizar2011phase}). \mod{A more recent and efficient approach is to use the refractive model~\cite{wittwer2022phase,aidukas2024high}, which directly reconstructs the object's refractive index, eliminating the need for phase unwrapping.}

\section{Numerical simulations}
In this section, we validate our approach through simulations. We compare reconstructions by the proposed method involving the probe reconstruction with the conventional approach involving division by the reference image. We also show the robustness of the proposed method to noise in measurements and demonstrate the convergence plots when solving the problem with the CG method described in the previous section. 
\begin{figure}[ht!]
\centering\includegraphics[width=1\textwidth]{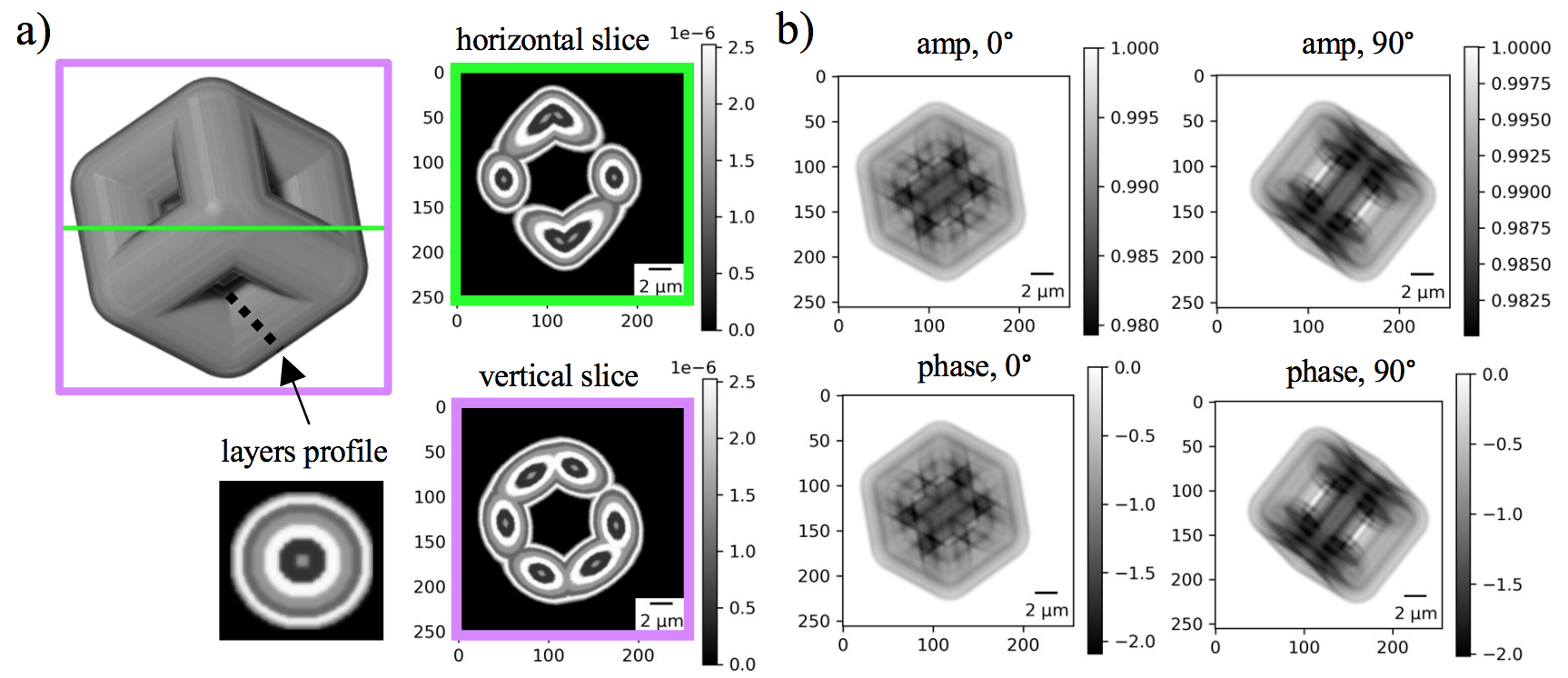}
\caption{Synthetic 3D cube object with layers of different phase component $\delta$ (a), amplitude and phase of the transmittance function for 0 and 90 degrees object rotation (b).}
\label{fig:3dsyn}
\end{figure}

For simulations, we constructed a 3D sample which mimic a cubic structure of a real ALD standard sample containing thin layers of different refractive index, see Figure~\ref{fig:3dsyn}a. The figure also shows examples of horizontal and vertical slices through the sample, marked with the green and blue colors, respectively, and a layers profile marked with red color. The colorbars are associated with the phase component $\delta$, the absorption component $\beta$ is 100 times lower. The whole volume size is $256\times256\times256$, where each voxel has a size of \SI{80}{\nano\metre} in each dimension. The transmittance function was computed as $\psi_k = e^{\frac{2\pi i}{\nu} \Rop_k u}$ for a wavelength of $\nu=\SI{7.27183e-11}{\metre}$, corresponding to X-ray energy of \SI{17.05}{\keV}. In Figure~\ref{fig:3dsyn}b, we illustrate examples of the amplitude and phase of the transmittance function for angles $\theta_k=0^\circ$ and $\theta_k=90^\circ$. Note that the sample exhibits very low absorption contrast (the amplitudes are close to 1), while the phase varies across a broader range, spanning from -2.0 to 0.0.

For modeling holography data on the detector we employed the PXM instrument settings from beamline ID16A at the ESRF. The settings are listed in Table~\ref{tab:pars}. In practice, $2\times2$ detector data binning is applied, resulting in a pixel size of \SI{3}{\micro\metre} after a 10-fold objective magnification.
For simulations with smaller data sizes $(256\times256$), we increased the pixel size to $3\times(2048/256)\,\SI{}{\micro\metre} = \SI{24}{\micro\metre}$, i.e. modelling $16\times16$ detector data binning. Tomography data were simulated at four sample planes for 180 angles, covering the angular span \mod{from 0 degrees to 180 degrees}.

\begin{table}[ht!]
    \centering
    \small{
    \begin{tabular}{|c|c|c|c|}
    \hline
         Energy& \SI{17.05}{\keV} \\
         Detector pixel size & \SI{15}{\micro\meter} \\
         Objective & 10$\times$ \\
         Detector size & 4096$\times$4096 px \\
         Detector size after binning & 256$\times$256 px (synthetic data),\\& 2048$\times$2048 px (experimental data) \\
         Focal spot to detector distance $Z$ & \SI{1.208}{\metre} \\          
         Number of sample planes & 4 \\
         Focal spot to sample distances $(z_j)$ & $4.026, 4.199, 4.890, \SI{6.325}{\milli\metre}$ \\         
         Corresponding geometric magnifications & $300,\, 287.7,\, 247,\, 191$\\          
         \hline
    \end{tabular}}
    \caption{The PXM instrument settings utilized for both simulating results and acquiring data at the beamline. }
    \label{tab:pars}
\end{table}

For simulations, we utilized the real probe function $q$ previously recovered through the near-field ptychography method~\cite{stockmar2013near}. Figure~\ref{fig:3dsyndata}a shows the amplitude and phase of the probe. Note the structure of the probe contains the horizontal and vertical features from the multilayer coated KB mirrors. In Figure~\ref{fig:3dsyndata}b we give examples of synthetically generated data using formula~\eqref{Eq:fwd_model_cone3}, while Figure~\ref{fig:3dsyndata}c demonstrates images for the ${d_j}/{d^r}$ term typically used during conventional reconstruction with formulation~\eqref{Eq:fwd_model_cone2}.  The image distortion is evident in the form of horizontal and vertical line artifacts, especially noticeable in zoomed-in regions.

\begin{figure}[ht!]
\centering\includegraphics[width=1\textwidth]{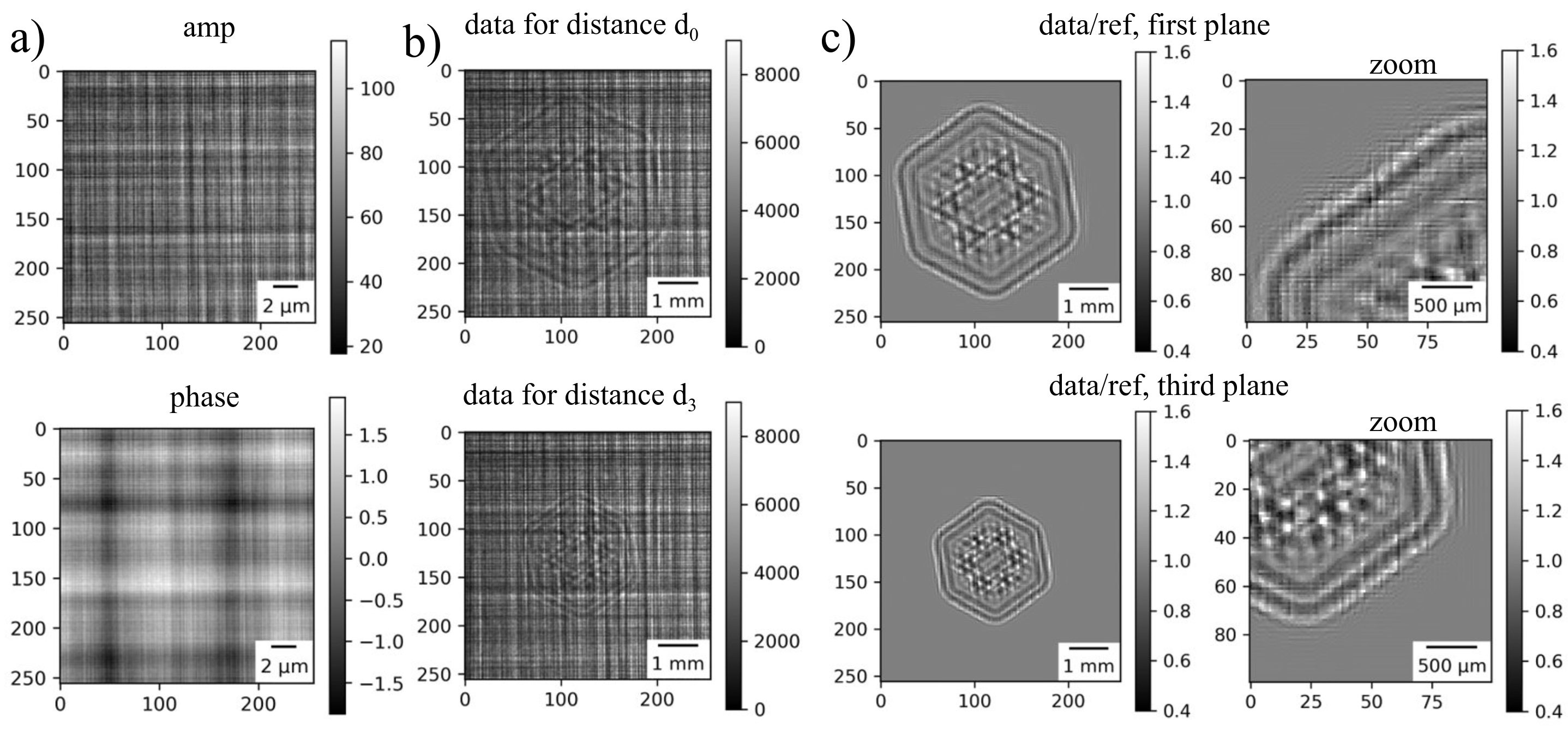}
\caption{Example of the real probe function from the KB-mirrors, recovered at distance $d_0$ (closest) from the focal spot (a), simulated data for the sample placed at shorter $(d_0)$ and longer $(d_3)$ distances from the focal spot (b), and (c) - the same data divided by the reference image, with zoomed regions demonstrating the artifacts.}
\label{fig:3dsyndata}
\end{figure}

\begin{figure}[ht!]
\centering\includegraphics[width=1\textwidth]{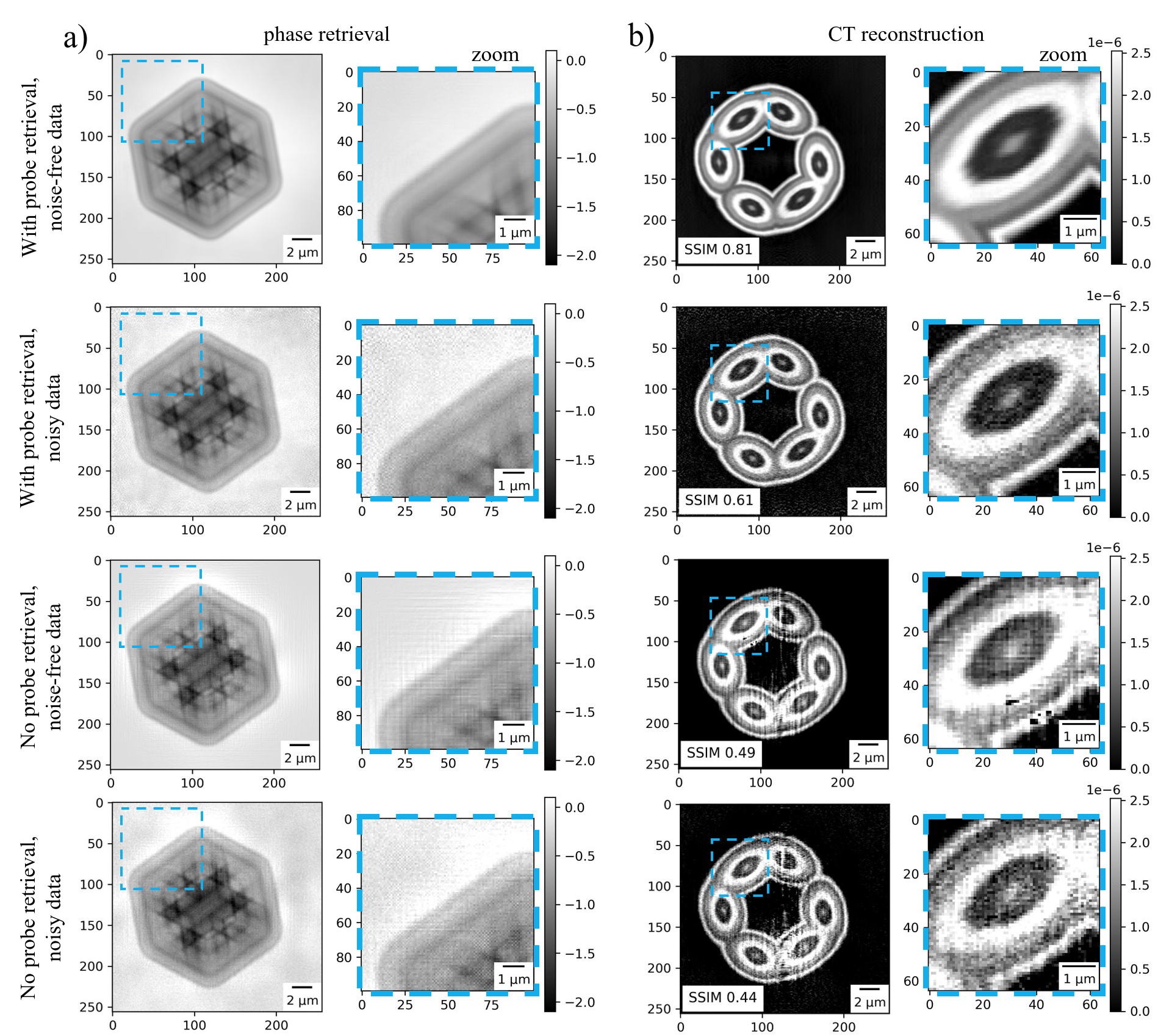}
\caption{Phase retrieval reconstruction using the proposed approach with probe retrieval and using the conventional approach with the division by the reference image, for both noise-free and noisy synthetic data (a); vertical slices through corresponding CT reconstructed volumes (b).}
\label{fig:3dsynrec}
\end{figure}

In Figure~\ref{fig:3dsynrec}, we present the reconstruction results, comparing the proposed approach, which simultaneously retrieves the sample transmittance function and probe function \mod{(minimizing functional \eqref{Eq:min1})}, with the conventional approach involving division by the reference image. \mod{The minimization functional for the conventional approach \eqref{Eq:fwd_model_cone2} naturally reads, upon adding sample shifts, as follows
\begin{equation}\label{Eq:fwd_model_cone2c}
  \tilde{F}_k(\psi_k) = \sum_{j=0}^{N_z-1}\left\|\left| \Pop_{\zeta_j}\big( \Sop_{t_{k,j}}(\psi_k)\big) \right|^2 -  \Mop_{1/m_j}\left(\sqrt{{d_{k,j}}/{d^r}}\right)\right\|_2^2.
\end{equation}
}
\mod{We remark that, in contrast to \eqref{Eq:min1} where the variable $q$ is shared across different projections resulting in a single minimization functional, the functionals in \eqref{Eq:fwd_model_cone2c} for $k=0,\dots,N_{\theta}-1$
are independent of each other and can be minimized separately using the CG method. }
 
Additionally, we evaluate the methods' robustness to Poisson noise in the data. Here, Poisson noise was introduced after dividing values by 20, mimicking low photon counts on the detector. Subsequently, the values were multiplied back by 20 to maintain consistency with the noise-free results.
Figure~\ref{fig:3dsynrec}a illustrates that phase retrieval by the proposed method significantly outperforms the conventional approach, both for noise-free and noisy data. 

In experimental holography data, the real part of the refractive index $\delta$, which signifies phase information, exhibits much higher contrast compared to the imaginary part $\beta$, which reflects sample absorption. Consequently, in tomography, the focus typically lies on solving the tomography problem only for the phase component. Corresponding CT reconstruction of $\delta$ obtained with \eqref{Eq:tomo} formulation confirms the improvement by the proposed method, as shown in Figure~\ref{fig:3dsynrec}b. We see that artifacts observed in the phase-retrieved reconstructions by the conventional approach are further amplified in the CT reconstructions.
For numerical validation, we calculated the Structural Similarity Index (SSIM) between the CT reconstructions and the ground truth object. SSIM values are depicted in the images, and notably, they are much higher for the proposed method. \mod{Notably, the SSIM index for the reconstruction of noisy data using the proposed method (0.61) surpasses that of the noise-free data reconstruction (0.49), highlighting the effectiveness and robustness of the proposed approach.}

\begin{figure}[ht!]
\centering\includegraphics[width=1\textwidth]{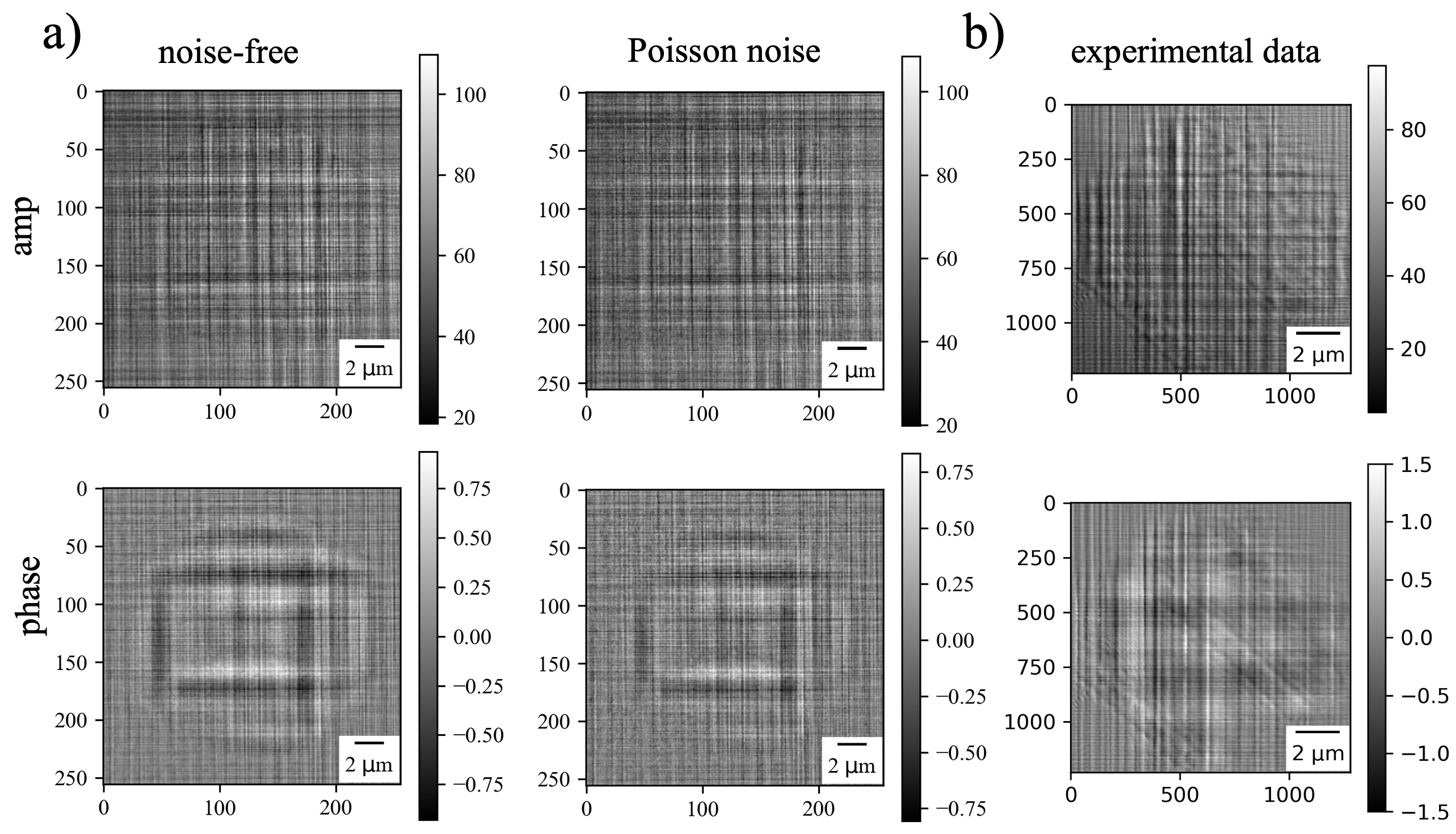}
\caption{Reconstructed probe function with the proposed iterative phase retrieval reconstruction: a) for noise-free and noisy synthetic data, b) for 3D ALD standard sample data acquired at beamline ID16A at the ESRF.}
\label{fig:3dsynrecill}
\end{figure}

The reconstructed amplitudes and phases of the probe by the proposed method are depicted in Figure~\ref{fig:3dsynrecill}a. The reconstructed probe is not accurate when compared to the ground truth shown in Figure~\ref{fig:3dsyndata}a. Notably, there is a lack of reconstruction at the corners, as neither projection of the sample during rotation passes through these regions. More accurate reconstruction can be obtained for samples covering the whole field of view. However, it is important to note that precise probe retrieval is not necessary. Rather, the goal is to find a probe that ensures data consistency to the greatest extent possible.

To obtain the phase retrieval reconstructions depicted in Figure~\ref{fig:3dsynrec}, we employed 1000 iterations of the CG method for both noise-free and noisy data cases within the proposed approach including probe retrieval. As an initial guess for the sample transmittance functions $\psi_k, k=0,\dots,N_\theta-1$, we utilized the reconstruction generated by the Multi-Paganin method, while the initial probe estimate was computed by propagating the reference image back to the sample plane 0. Reconstruction by the conventional approach with division by the reference image was also first done by the Multi-Paganin method, followed by 100 CG iterations for solving \eqref{Eq:fwd_model_cone2} to adjust high frequencies in reconstructions. No convergence was observed after 100 iterations.

\begin{figure}[ht!]
\centering\includegraphics[width=1\textwidth]{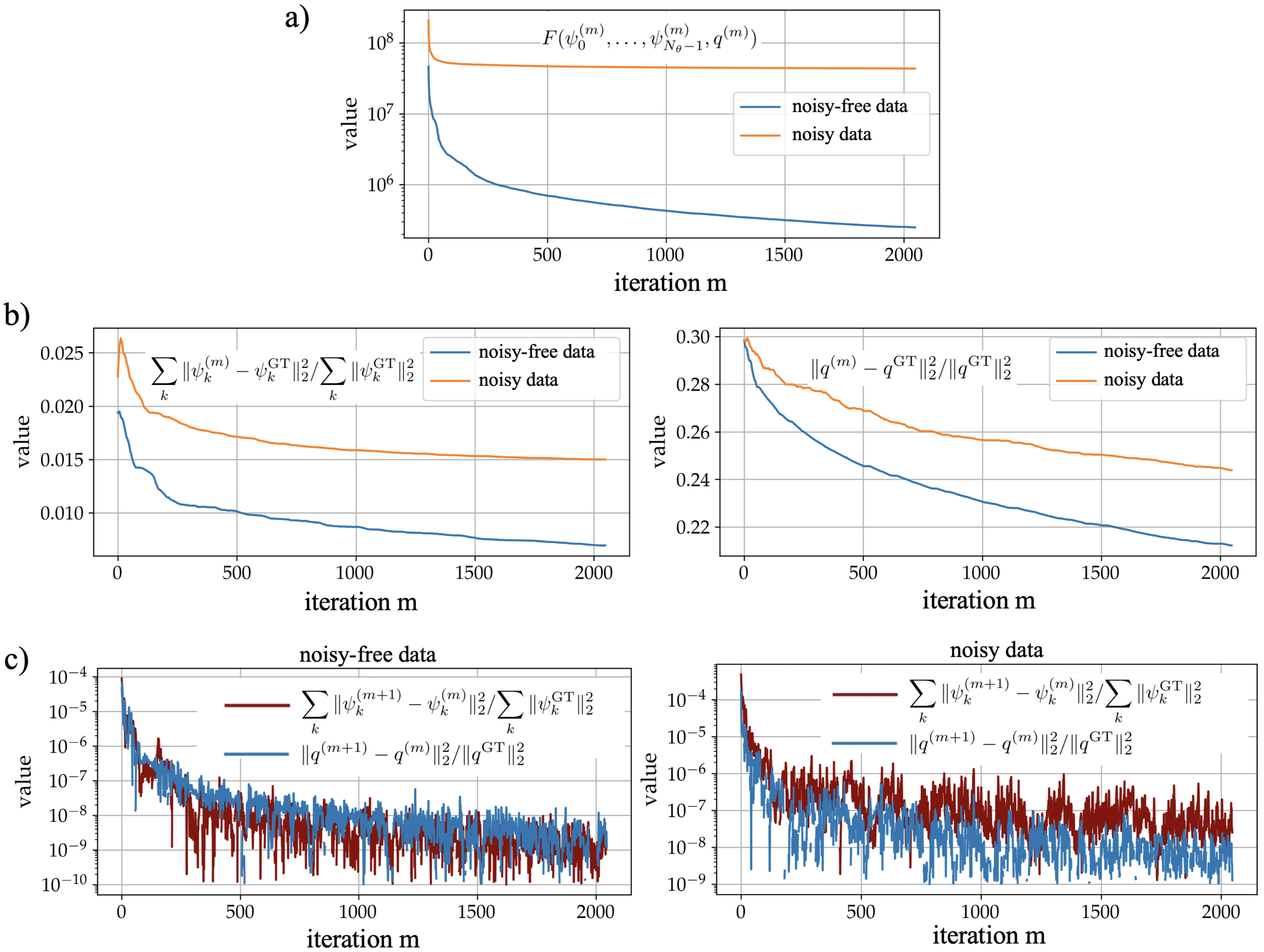}
\caption{Convergence analysis for the phase retrieval reconstruction with probe retrieval: a) \mod{convergence of the objective function \eqref{Eq:min1}}; b) relative errors compared to the ground truth (GT) for the object and probe; c) relative differences between the objects and probe at $m$ and $m+1$ iterations for noise-free and noisy data cases. Gaps in the plots correspond to no difference between iterations.}
\label{fig:conv}
\end{figure}

In Figure~\ref{fig:conv}, we investigate the convergence behavior of the proposed method. The convergence of the minimization functional \eqref{Eq:min2} is illustrated in Figure~\ref{fig:conv}a. We observe a smooth decrease in the functional for noise-free data, wherein the decrease is assured by the fact that $F\left(\psi_k^{(m+1)},\dots,\psi_k^{(m+1)},q^{(m+1)}\right)\le F\left(\psi_k^{(m)},\dots,\psi_k^{(m)},q^{(m)}\right)$ at each iteration, thanks to the line search procedure. For the noisy data case, we observe that the minimization functional is not changed significantly after 100-200 iterations. This is due to the fact that there is a big difference between the data with noise and data modeled by propagating the object to the detector on each iteration. 
Examining the relative errors compared to the ground truth (GT) displayed in Figure~\ref{fig:conv}
b, we still observe a notable convergence rate even after 1000 iterations.  Additionally, in Figure~\ref{fig:conv}c, we present the relative differences between the object and probe at iterations $m$ and $m+1$. These plots display oscillatory patterns because the problem is concurrently solved for the object and probe using the CG method, with the gradient step direction chosen to expedite the convergence of the target minimization functional. 
At certain iterations, the object function remains unchanged, resulting in gaps in the plots. This indicates that the optimal object representation has been found for the current probe. Object updates resume after subsequent updates of the probe. A similar phenomenon of zero-update behavior also occurs for the probe function. These simultaneous updates of the object and illumination are advantageous for avoiding local minima.

Performing numerous CG iterations is time-consuming, especially when processing experimental data at full resolution. In addition to leveraging GPU accelerations within the developed software package, which provides GPU-accelerated operators, we have explored several optimization strategies to enhance the practical usability of the method.
Firstly, one optimization involves providing a more accurate initial guess for the probe, rather than relying solely on the calculation obtained by back-propagating the reference image. This initial guess can be sourced from previous experiments or from experiments where it was recovered using the Near-Field Ptychography method~\cite{stockmar2013near}. Subsequently, the iterative scheme adjusts this initial guess to better align with the current measurements, typically within 100-200 iterations. Secondly, the probe can be approximated by considering only a limited set of projections. In practice, we have found that using 20-50 projections can yield satisfactory results while also significantly reducing the total reconstruction time.
Thirdly, we observed that slow convergence often results from the gradual fitting of low frequencies in the object. To address this, we implemented a hierarchical reconstruction approach involving data binning and upsampling. In this approach, we perform a certain number of iterations on binned data and then extrapolate the recovered object with the probe function to a grid twice as dense. This result serves as an initial guess for the next set of iterations, where the data are less binned. The final object and probe functions are reconstructed on dense grids without binning.
As an illustrative example, for processing data from a $2048\times2048$ detector, one could initiate reconstruction with binning $8\times8$ for 1000 iterations. Subsequently, the reconstructions are upsampled, and iteration continues with data binned by $4\times4$ for 500 iterations, and so forth. Finally, reconstruction without binning can be completed with 125 iterations. This scheme yields reconstructions of the same quality and resolution as those obtained in full resolution without binning, which typically require more than 1000 iterations.

\section{Experimental data processing}

To demonstrate the efficiency of the proposed method in experimental data processing, we examined the reconstruction of a 3D standard sample fabricated using Atomic Layer Deposition (ALD)~\cite{becker2018atomic}. This sample was created using the Photonic Professional GT 3D printer from Nanoscribe, which allows for micro- and nanofabrication via two-photon polymerization. The ALD process involved depositing ZnO and Al$_2$O$_3$ layers of varying thickness (4-\SI{80}{\nano\metre}) onto a cubic pattern. In \cite{becker2018atomic}, the quality of the deposition was verified using  X-ray Reflectivity (XRR), Scanning Electron Microscopy (SEM) and Transmission X-ray Microscopy (TXM). Notably, TXM full-field imaging only allowed for separating layers of down to \SI{60}{\nano\metre} thickness due to the limited resolution defined by the zone plate.

The same 3D ALD standard sample was measured with the PXM instrument at ID16A beamline of the ESRF. The instrument settings are listed in Table~\ref{tab:pars}, examples of data acquired by the detector for different sample planes are shown in Figure~\ref{fig:setup}. Projection data for each sample plane were collected for 1500 angles spanned over 180 degrees. Exposure time per projection was \SI{0.1}{\s}. Pixel size after magnification is \SI{10}{\nano\metre}. The data acquisition was done in a step-scan mode, where the sample was randomly shifted for each projection and all shifts were recorded. This procedure helps in reducing ring artifacts in reconstructions by the conventional approach, however, it is not always necessary when the sample reconstruction is done together with the probe retrieval. 

\begin{figure}[ht!]
\centering\includegraphics[width=1\textwidth]{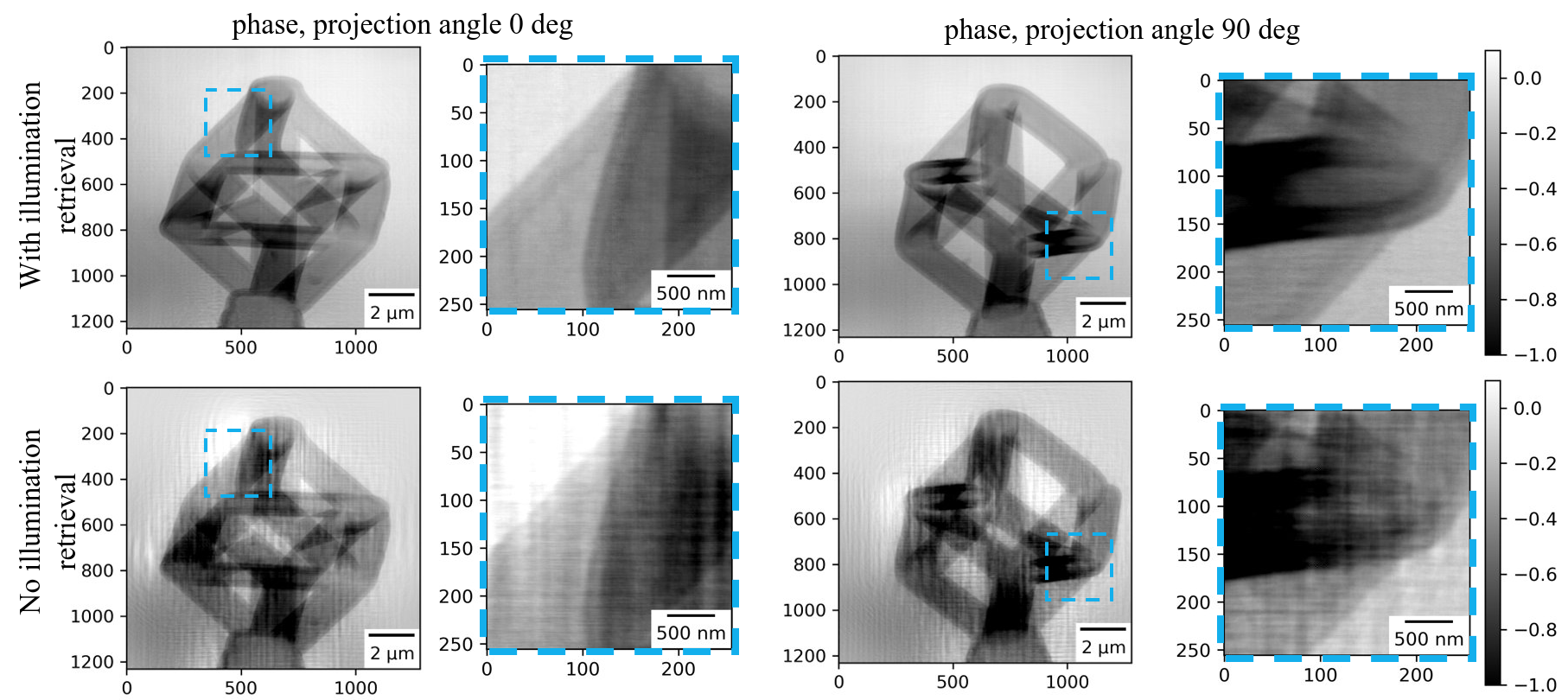}
\caption{Phase retrieval for the 3D ALD sample data measured at ID16A beamline by the proposed method with probe retrieval (first row), and by the conventional approach without probe retrieval (second row).}
\label{fig:recrealphase}
\end{figure}
Phase retrieval reconstructions using the proposed method were conducted independently in angular chunks consisting of 150 angles. This number of angles proved sufficient for obtaining accurate probe approximations while also optimizing computational resource utilization within a single computing node. This approach also facilitates easy parallelization across multiple GPUs and computing nodes. Additionally, we adopted a hierarchical reconstruction approach, as described in the previous section. Reconstruction was started with binning $8\times 8$ for 1000 iterations, and ended up without binning for 125 iterations. As for synthetic data in the previous section, we employed the reconstruction produced by the Multi-Paganin method as the initial guess for the sample transmittance functions $\psi_k, k=0,\dots,N_\theta-1$. For the probe function, the initial guess was the reference image propagated back to sample plane 0. The total reconstruction time on two computing nodes, each equipped with 4 Tesla A100 GPUs, was approximately 4 hours. The subsequent tomographic reconstruction involved 500 CG iterations that were completed for the whole volume within 15 minutes.

In the conventional reconstruction approach involving division by the reference image, we utilized the Multi-Paganin method and 100 CG iterations to solve \eqref{Eq:fwd_model_cone2} for fitting high frequency components.

Figure~\ref{fig:recrealphase} presents a comparison of phase retrieval reconstructions from two methods, illustrating the phase component of the transmittance function for angles $0^\circ$ and $90^\circ$. The significant improvement achieved with the new method is evident. This enhancement is further demonstrated in the 3D CT reconstructed volume, specifically illustrating \mod{the real part, $\delta$, of the complex refractive index}. The reconstructed amplitudes and phases of the probe function obtained using the proposed method are shown in Figure~\ref{fig:3dsynrecill}b. \mod{Similar to the synthetic data results in the previous section, the probe function is not reconstructed at the corners, as no projection of the sample during rotation passes through these regions.} The recovered probe ensures data consistency between different sample planes and projection angles.

\begin{figure}[ht!]
\centering\includegraphics[width=1\textwidth]{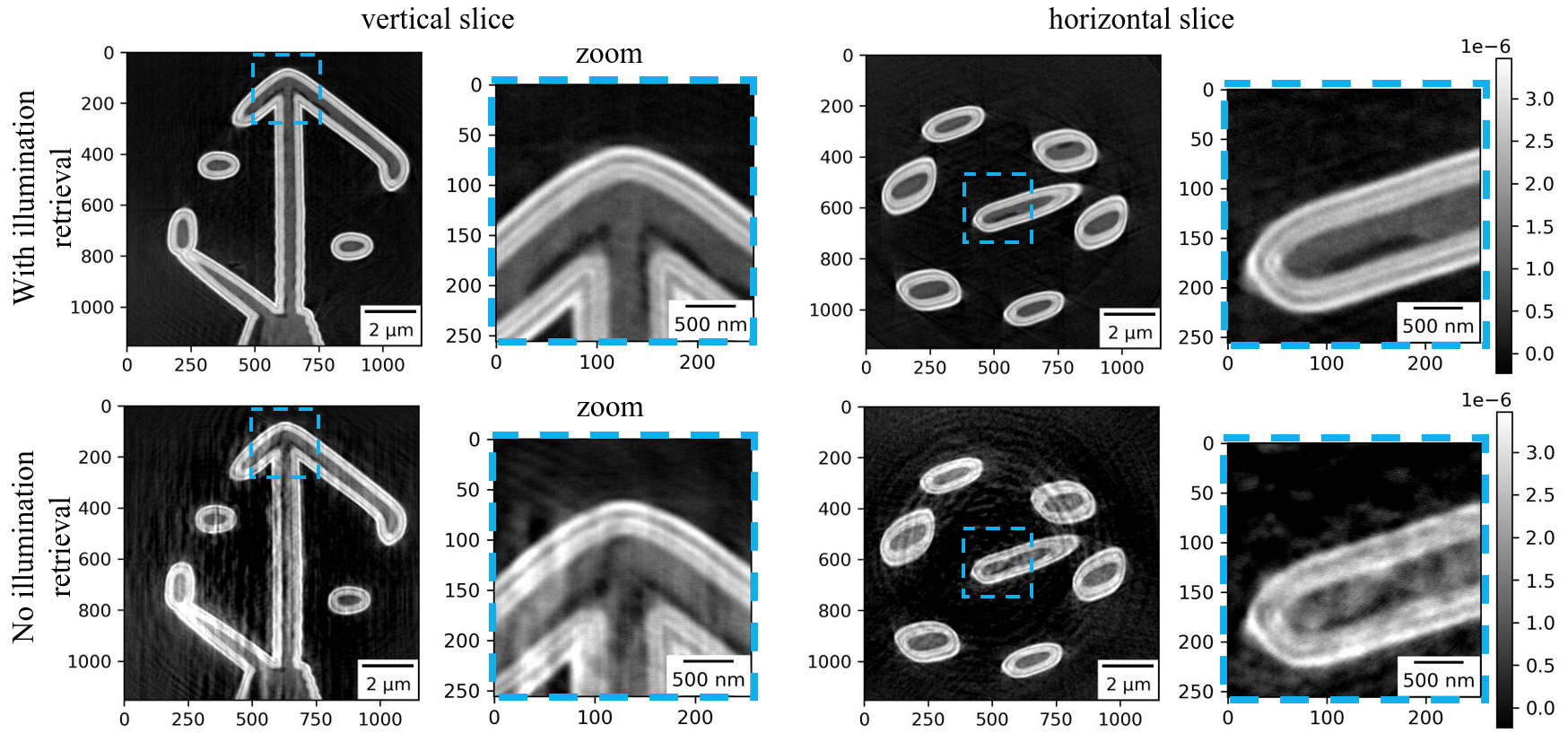}
\caption{Vertical and horizontal slices through CT reconstructed volumes of the 3D ALD sample, applied after phase reconstruction with probe retrieval (first row), and after phase reconstruction without probe retrieval (second row). }
\label{fig:recreal3D}
\end{figure}

The improved quality and resolution obtained with the new method enable a more precise analysis of the reconstructed volume. In Figure~\ref{fig:recrealanalysis}, we present a 3D rendering of the sample in ORS Dragonfly package, along with a line profile through one of the cube borders. This border comprises 10 ALD alternating layers of ZnO and Al$_2$O$_3$ compounds, with the outer layers reaching thicknesses of up to \SI{80}{\nano\metre} and the inner layers as thin as \SI{5}{\nano\metre}. The tabulated values for the phase component $\delta$ of the refractive index at \SI{17.05}{\keV} energy for ZnO and Al$_2$O$_3$ are \SI{3.763e-06}{} and \SI{2.805e-06}{}, respectively.
The targeted and measured thicknesses of the layers are listed in the table on the right side of the figure, sourced from \cite{becker2018atomic}, where the measured values were obtained using the X-ray reflectivity (XRR) method. The line profile through a cube border reconstructed with the new method clearly demonstrates the separation of layers c7-c10, with thicknesses as low as \SI{40}{\nano\metre}, while layers c1-c6 with smaller thicknesses appear merged into a single feature. In contrast, with the conventional method, only layers c9-c10 with a thickness of \SI{80}{\nano\metre} can be separated. It is important to note that in some regions, the artifacts when using the conventional method are so strong that it becomes impossible to separate any layer.

\begin{figure}[ht!]
\centering\includegraphics[width=1\textwidth]{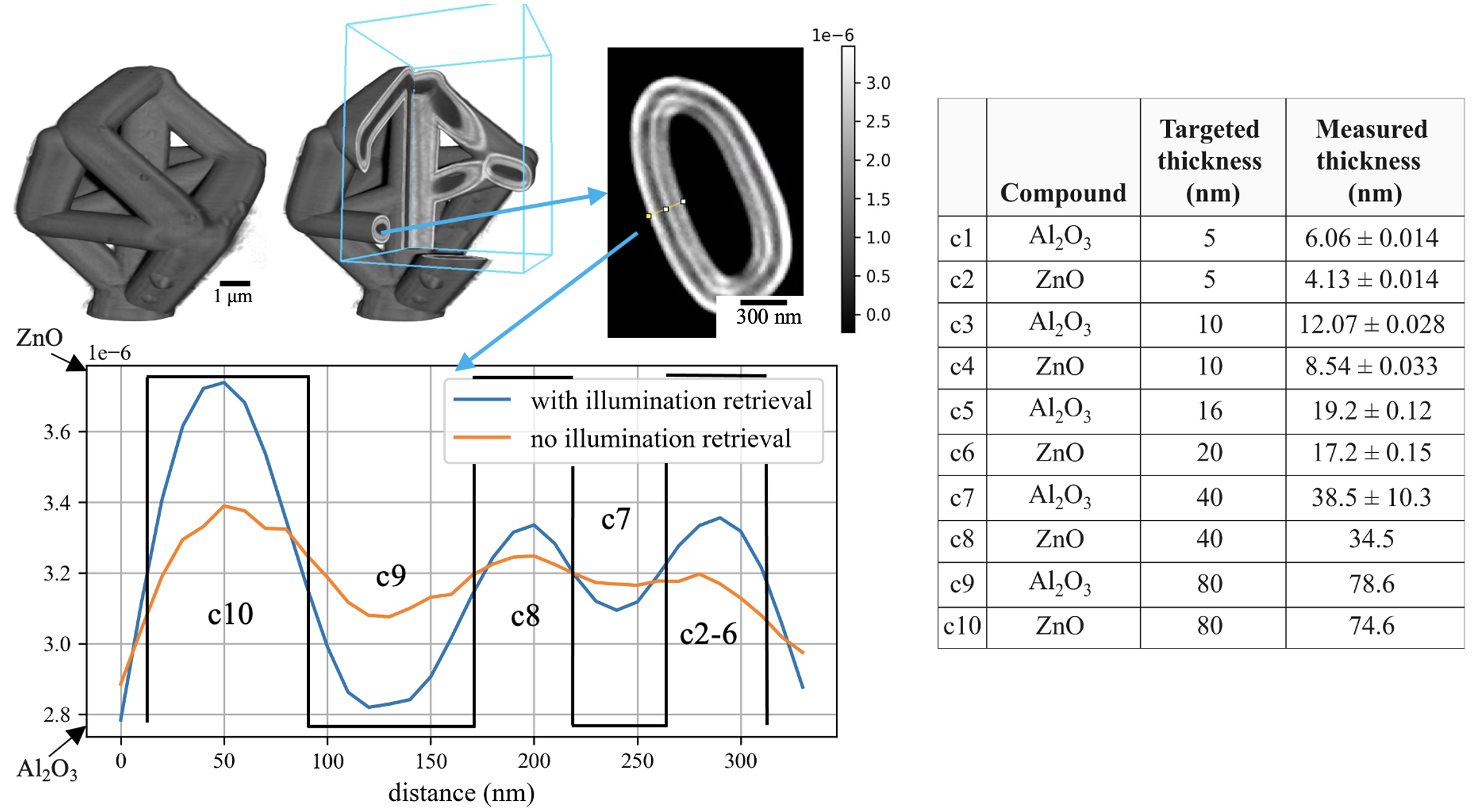}
\caption{3D rendering of the 3D ALD sample reconstruction, estimation of the layers thickness, and comparing the thickness to the ones provided in Table II from \cite{becker2018atomic}. The black line indicates the theoretical profile of the layers based on the table.}
\label{fig:recrealanalysis}
\end{figure}

\section{Conclusions and outlook}
In this paper, we demonstrated the efficacy of combining phase retrieval reconstruction with probe retrieval to achieve higher quality and resolution results in holotomography data. 

The reconstruction scheme is implemented in our Python-based HolotomocuPy package publicly available at \url{https://holotomocupy.readthedocs.io/}. The package features GPU-accelerated implementations of basic operators for wavefront propagation, data alignment, and tomographic data processing. Regular linear algebra operations were accelerated using the cuPy library. Additionally, more complex functionalities are implemented as CUDA kernels called through cuPy. The package offers a Jupyter notebook interface, enabling convenient implementation of new methods using the existing basic operators and adding acquisition specific or domain functionality. 

Synthetic data results, utilizing real probe function from the KB mirrors, demonstrated a significant improvement in the final 3D reconstruction of the sample volume, with an SSIM index of 0.81 compared to 0.49 for the conventional method. Moreover, the method exhibited superior performance and robustness to Poisson noise in experimental data, with an SSIM index of 0.61 compared to 0.44 for the conventional method. We also provided guidance on optimizing the number of iterations and convergence rate.

For experimental data acquired with the PXM instrument at beamline ID16A of the ESRF, we showcased the capability to resolve features as small as \SI{40}{\nano\metre} using the new approach. This represents a notable enhancement over previous full-field X-ray imaging techniques, such as the Transmission X-ray Microscope, which could only resolve larger features of \SI{80}{\nano\metre}~\cite{becker2018atomic}. The thickness of the resolved layers aligned with the values presented in Table II from~\cite{becker2018atomic}. Additionally, the recovered real component $\delta$ of the complex refractive index for ZnO and Al$_2$O$_3$ layers corresponded to the tabulated values.

While our method has demonstrated remarkable results, there are still avenues for further improvement. In reconstruction, we assumed that the probe function remains constant during data acquisition, which is a valid assumption for synchrotrons of new generation and well-optimized instruments like the ID16A beamline. However, potential sources of probe variation, such as changes in the ring filling, beam and optics stability, or sample stability, could still impact results. In this case, additional pre-processing procedures could be applied to identify corresponding shifts or multipliers and incorporate them into the iterative schemes as new operators.
Another potential improvement lies in jointly solving the holography and tomography problems. The identity in \eqref{Eq:fwd_model_cone3} is then modified as
\begin{equation}\label{Eq:fwd_model_cone3_future}
  \left|\frac{z_0}{ Z}\Mop_{m_0}\left(\Pop_{\zeta_j/\tilde m_j^2}\left(\Pop_{\omega_j}(q)\cdot\Mop_{1/\tilde m_j}\left(e^{\frac{2\pi i}{\nu} \Rop_k (u)}\right) \right)\right)\right|^2 = d_{j,k}, 
\end{equation}
and solved with respect to the object refractive index $u$ and probe $q$.
This approach enables solving the entire reconstruction problem using the Alternating Direction Method of Multipliers (ADMM), which splits the problem into local holography and tomography subproblems. These subproblems are then coordinated through dual variables to find a solution for the original problem. This scheme improves convergence rates and allows for reducing the number of measurements needed to achieve high-quality and high-resolution results. Similar techniques have been successfully applied in related fields, such as 3D ptychography problems~\cite{Aslan:19,nikitin2019photon} and tomography problems with nonrigid alignment~\cite{nikitin2021distributed}. \mod{The efficient phase unwrapping method based on the refractive model from ~\cite{wittwer2022phase} will also be adapted for the proposed scheme and implemented in HolotomocuPy package.}
\mod{Another optimization in reconstruction involves automating the adjustments of sample-to-focal spot distances $z_j$ and probe/sample shifts $s_{k,j},~t_{k,j}$. Accurately determining these distances in practice is challenging without dedicated measurements, as the translational and rotational alignments of instruments, particularly the KB mirrors, are continually affected by thermal loads and cumulative position drifts. Our future work will focus on integrating position correction directly into the iterative reconstruction process. We also plan to generalize this technique for holotomography problems with nonrigid projection data alignment.}

\appendix
\section{Appendix A}

\subsection{Cone-beam propagation, ideal case}\label{cbp}

In this section we first review the Fresnel scaling theorem, for completeness, and then proceed to derive the formula for propagating both probe and sample to the detector, in the cone-beam setting. We remind the reader that, in the near-field approximation, a field $\psi(\boldsymbol{x})$ at $z=0$ is propagated to a plane at $z=Z$ via convolution with the Fresnel kernel \begin{equation}\label{defPz}
P_Z(\boldsymbol{x})=\frac{e^{ikZ}}{i\lambda Z}\exp\left(i\pi|\boldsymbol{x}|^2/\lambda Z\right)
\end{equation}
where $\boldsymbol{x}=(x_1,x_2)$, $\lambda$ is the wavelength and $k=2\pi/\lambda$ (see e.g.~equation (1.41) \cite{paganin2006coherent}). In the notation of the manuscript we thus have \begin{equation}\label{defGz}
    \Pop_Z(\psi)=\psi * P_Z.
\end{equation}
Physically, the only difference between modeling a cone beam and a parallel beam is the shape of the incoming field.  Given a perfect monochromatic coherent cone beam emanating at the origin, the field in the object plane at $z_j$ is given by the amplitude of the incoming beam times 
$\Pop_{z_j}\big(\delta_0\big)=  P_{z_j}
$ (before hitting the object), where  $\delta_0$ denotes the Dirac distribution.  

After passing the object $\psi$, the total field is thus $P_{z_j}\cdot \psi $ which is propagated to $Z$ by convolution with $P_{Z-z_j}$, giving rise to the field  $\Pop_{Z-z_j}\left( P_{z_j}\psi \right)$ on the detector. Letting $\simeq$ denote that two functions are the same except for a unimodular factor, the Fresnel scaling theorem reads as follows.

\begin{theorem}\label{t2}
 The field on the detector is given by
\begin{equation}\label{e1}
    \Pop_{Z-z_j}\left(P_{z_j}\cdot \psi \right)(\boldsymbol{x})\simeq P_Z(\boldsymbol{x})\Pop_{\zeta_j}(\psi)\left( \boldsymbol{x}/m_j\right)
\end{equation}
where $m_j=Z/z_j$ is the so called ``magnification''-factor and $\zeta_j=(Z-z_j)/m_j$ is the modified distance.
\end{theorem}
We remark that the above theorem looks slightly different from the presentation in the standard reference \cite{paganin2006coherent} (Appendix B). This is because the (unimodular) factor $P_Z$ can be ignored if one is only interested in intensities. However, for the treatment here it is crucial to keep this factor, since in the next section we shall use the formula to propagate between multiple planes.

\begin{proof}
We have \begin{equation}\label{e2}\begin{aligned}
    & \Pop_{Z-z_j}\left(P_{z_j}\cdot \psi \right)(\boldsymbol{x}) \simeq\\&\frac{1}{\lambda^2 z_j(Z-z_j)}\int \psi(\boldsymbol{w})\exp\left(i\pi|\boldsymbol{w}|^2/\lambda z_j\right)\exp\left(i\pi|\boldsymbol{x}-\boldsymbol{w}|^2/\lambda (Z-z_j)\right)~d\boldsymbol{w}
\end{aligned}\end{equation}
where the exponential functions can be rewritten
\begin{align*}
    &\exp\left(i\pi|\boldsymbol{w}|^2/\lambda z_j\right)\exp\left(i\pi|\boldsymbol{x}-\boldsymbol{w}|^2/\lambda (Z-z_j)\right)=\\&
    \exp\left(\frac{i\pi}{\lambda}\left(\left(\frac{1}{z_j}+\frac{1}{Z-z_j}\right)|\boldsymbol{w}|^2-\frac{2}{Z-z_j}\boldsymbol{x}\cdot\boldsymbol{w}+\frac{1}{Z-z_j}|\boldsymbol{x}|^2 \right)\right)=\\&
    \exp\left(\frac{i\pi}{\lambda}\left(\left(\frac{1}{z_j}+\frac{1}{Z-z_j}\right)|\boldsymbol{w}-\frac{1/(Z-z_j)}{1/z_j+1/(Z-z_j)}\boldsymbol{x}|^2 \right)\right)\cdot\\&\hspace{2cm}\exp\left(\frac{i\pi}{\lambda}\left(\left(\frac{1}{Z-z_j}-\frac{1/(Z-z_j)^2}{1/z_j+1/(Z-z_j)} \right)|\boldsymbol{x}|^2 \right)\right)
\end{align*}
Now note that $1/\zeta_{j}=1/z_j+1/(Z-z_j)$, that $$\frac{1/(Z-z_j)}{1/z_j+1/(Z-z_j)}=\frac{z_j}{Z}$$ and $$\frac{1}{Z-z_j}-\frac{1/(Z-z_j)^2}{1/z_j+1/(Z-z_j)}=\frac{1}{Z-z_j}\left(1-\frac{z_j}{Z}\right)=\frac{1}{Z}$$
to rewrite the above as
\begin{equation}\label{e3}
    \exp\left(\frac{i\pi}{\lambda}\left(\frac{1}{\zeta_{j}}|\boldsymbol{w}-\frac{z_j}{Z}\boldsymbol{x}|^2 \right)\right)\exp\left(\frac{i\pi}{\lambda}\left(\frac{|\boldsymbol{x}|^2}{Z} \right)\right)\simeq \lambda\zeta_j  P_{\zeta_{j}}\left(\boldsymbol{w}-\frac{z_j}{Z}\boldsymbol{x}\right)\exp\left(i\pi|\boldsymbol{x}|^2/\lambda Z\right).
\end{equation}
Inserting this in \eqref{e2} and canceling factors we get
$$ \simeq \frac{\zeta_j}{\lambda z_j(Z-z_j)}\int \psi(\boldsymbol{w})P_{\zeta_{j}}\left(\boldsymbol{w}-\frac{z_j}{Z}\boldsymbol{x}\right)\exp\left(i\pi|\boldsymbol{x}|^2/\lambda Z\right)~d\boldsymbol{w}$$
Upon noting that $\frac{\zeta_j}{z_j(Z-z_j)}=\frac{1}{Z}$, the above boils down to \eqref{e1}, and the proof is complete.

\end{proof}

Given $m\in \mathbb{R}$ let us denote by $\Mop_m$ the dilation (or magnification) operator $\Mop_{m}(\psi)(\boldsymbol{x})=\psi\left({\boldsymbol{x}}/{m}\right)$. The Fresnel scaling formula can then be written more compactly as \begin{equation}\label{e65}
\Pop_{Z-z_j}\left(P_{z_j}\cdot\psi \right)\simeq P_Z\cdot \Mop_{m_j}\big(\Pop_{\zeta_j}(\psi)\big).\end{equation}
We remark that, in case the object is inserted at some distance $z_j$ before the focal spot, the above derivation still works. In this case, both $m_j$ and $\zeta_j$ in formula \eqref{e65} become negative. This has the curious interpretation that instead of forward propagating 
$P_{z_j}\cdot \psi$ to the other side by $\Pop_{Z-z_j}$, we may \textit{backward-propagate} $\psi$ by the negative distance $\zeta$, then rescale it by the positive number $|m_j|$, and finally flip the image so that up becomes down and left becomes right. This will be further studied in upcoming works. 
The following result, which allows us to swap the order of the $\Mop$ and $\Pop$-operators, is also convenient.

\begin{proposition}\label{p1}
Given any $m,\zeta$ we have $\Mop_{m}\big(\Pop_\zeta(\psi)\big)\simeq \Pop_{\zeta m^2}\big(\Mop_m(\psi)\big)$.
\end{proposition}
\begin{proof}
By the change of variables $\boldsymbol{w}=\boldsymbol{t}/m$ we get
\begin{align*}
    &\Mop_{m}(\Pop_\zeta(\psi))(\boldsymbol{x}) \simeq \frac{1}{\lambda \zeta}\int  \psi(\boldsymbol{w})\exp\left(i\pi|\boldsymbol{x}/m-\boldsymbol{w}|^2/\lambda \zeta\right)~d\boldsymbol{w}=\\
    &\frac{1}{\lambda \zeta m^2} \int \psi(\boldsymbol{t}/m)\exp\left(i\pi|\boldsymbol{x}-\boldsymbol{t}|^2/\lambda \zeta m^2\right)~{d\boldsymbol{w}}\simeq \Pop_{\zeta m^2}(\Mop_m(\psi))(\boldsymbol{x}),
\end{align*}
as desired.
\end{proof}

In particular, with $\zeta_j$ and $m_j$ as above we have $\zeta_jm_j^2=(Z-z_j)m_j$ and hence \eqref{e65} can be recast as 
\begin{equation}\label{e11}
    \Pop_{Z-z_j}\left(P_{z_j}\cdot  \psi \right)\simeq P_Z\cdot \Pop_{(Z-z_j)m_j}\left(\Mop_{m_j}(\psi)\right).
\end{equation}

\subsection{Cone-beam propagation with probe}

We now tackle the case when the probe is abberrated e.g. by KB-mirrors or other systematic artefacts of the measurement setup. The probe in the plane $z_{0}$ (i.e.~the first measurement distance, see Figure~\ref{fig:setup}), is now modelled by $$\exp(\pi i |\boldsymbol{x}|^2/\lambda z_0)\cdot q\simeq \lambda z_0 P_{z_0}\cdot q,$$ where $q$ is a function in space that describes the probe at $z_0$ if we had been in the parallel beam geometry, containing both phase and amplitude. We thus absorb the intensity into $q$, so that $\|q\|^2$ equals the total intensity of the incoming beam at $z_0$. To forward model an object $\psi$ in the plane $z_0$ to the detector we simply apply $\mathcal{P}_{Z-z_{0}}$ to $\exp(\pi i |\boldsymbol{x}|^2/\lambda z_0)\cdot q \cdot\psi$, which easily can be computed via Theorem \ref{t2}. However, the situation gets more tricky after the sample has been moved to the other distances $z_j$, $j=1, 2,3$. The field on the detector is given by  
\begin{equation}\label{sista}\Pop_{Z-z_j}\left(\Pop_{z_j-z_0}\left(\exp(\pi i |\boldsymbol{x}|^2/\lambda z_0)\cdot q\right)\cdot\psi\right).\end{equation}
For practical purposes, the above expression needs to be transformed into something more implementation-friendly. Our main theorem reads

\begin{theorem}\label{t3}
The intensity on the detector, i.e.~the square modulus of \eqref{sista}, is given by
\begin{equation*}
\frac{z_0^2}{ Z^2}\left|\Mop_{m_0}\left(\Pop_{\zeta_j/\tilde m_j^2}\left(\Pop_{(z_j-z_0)/\tilde m_j}(q)\cdot \Mop_{1/\tilde m_j}(\psi)\right) \right)\right|^2,
\end{equation*}
where $\zeta_j={z_j(Z-z_j)}/{Z}$, $\tilde m_j={z_j}/{z_0}$ and $m_0={Z}/{z_0}$.
\end{theorem}

\begin{proof}
For readability we shall work with $P_{z_0}$ instead of $\exp(\pi i |\boldsymbol{x}|^2/\lambda z_0)$ (since they are the same except for a multiplicative constant). Applying formula \eqref{e65} to $\Pop_{z_j-z_0}\left( P_{z_0}\cdot q\right)$ (with $\tilde Z=z_j$, $\tilde m_j={z_j}/{z_0}$ and $\tilde \zeta_j={(z_j-z_0)}/{\tilde m_j}$) gives $$\Pop_{z_j-z_0}\left( P_{z_0}\cdot q\right)\simeq P_{z_j}\cdot \Mop_{\tilde m_j}(\Pop_{(z_j-z_0)/\tilde m_j}(q)).$$
    This expression is then multiplied by $\psi$ whereby we apply $\Pop_{Z-z_j}$. To simplify this expression we may use \eqref{e65} once more (now with $Z,~m_j$ and $\zeta_j$ as in Theorem \ref{t2}) to get \begin{equation}\label{e66}
        \begin{aligned}
        &\Pop_{Z-z_j}\left(P_{z_j}\cdot\Mop_{\tilde m_j}(\Pop_{(z_j-z_0)/\tilde m_j}(q))\cdot \psi \right)\simeq \\&P_{Z}\cdot\Mop_{m_j}\left(\Pop_{\zeta_j}\left(\Mop_{\tilde m_j}(\Pop_{(z_j-z_0)/\tilde m_j}(q))\cdot \psi \right)\right) \simeq\\&P_{Z}\cdot\Mop_{m_j}\left(\Pop_{\zeta_j}\left(\Mop_{\tilde m_j}\left(\Pop_{(z_j-z_0)/\tilde m_j}(q)\cdot \Mop_{1/\tilde m_j}(\psi) \right)\right)\right),  
    \end{aligned}
    \end{equation}
    where we used that $\Mop_{\tilde m_j}(\varphi_1)\cdot \varphi_2=\Mop_{\tilde m_j}\big(\varphi_1\cdot \Mop_{1/\tilde m_j}(\varphi_2)\big)$. We now apply Proposition \ref{p1} to the operator  $\Mop_{m_j}\Pop_{\zeta_j}\Mop_{\tilde m_j}$ which yields $$\Mop_{m_j}\Pop_{\zeta_j}\Mop_{\tilde m_j}=\Mop_{m_j}\Mop_{\tilde m_j}\Pop_{\zeta_j/\tilde m_j^2}=\Mop_{m_0}\Pop_{\zeta_j/\tilde m_j^2}.$$ 
By inserting this in \eqref{e66} and multiplying everything with $\lambda z_0$, we obtain \begin{align*}
  &\Pop_{Z-z_j}\left(\Pop_{z_j-z_0}\left(\exp(\pi i |\boldsymbol{x}|^2/\lambda z_0)\cdot q\right)\cdot\psi\right)\simeq \\&\lambda z_0 P_{Z}\Mop_{m_0}\left(\Pop_{\zeta_j/\tilde m_j^2}\left(\Pop_{(z_j-z_0)/\tilde m_j}(q)\cdot \Mop_{1/\tilde m_j}(\psi)\right) \right).
\end{align*}
The desired formula follows by applying the modulus, noting that $|P_{Z}|={1}/{\lambda Z}$.
\end{proof}

\section*{Acknowledgments}

This research used resources of the Advanced Photon Source, a U.S. Department of Energy (DOE) Office of Science user facility and is based on work supported by Laboratory Directed Research and Development (LDRD) funding from Argonne National Laboratory, provided by the Director, Office of Science, of the U.S. DOE under Contract No. DE-AC02-06CH11357.

\section*{Disclosures}
The authors declare no conflicts of interest.

\section*{Data availability}
Data underlying the results presented in this paper are not publicly available at this time but may be obtained from the authors upon reasonable request.

\bibliography{refs}

\begin{thebibliography}{10}
\newcommand{\enquote}[1]{``#1''}

\bibitem{cloetens1999holotomography}
P.~Cloetens, W.~Ludwig, J.~Baruchel, \emph{et~al.}, \enquote{Holotomography: Quantitative phase tomography with micrometer resolution using hard synchrotron radiation x rays,} {\protect\JournalTitle{Applied physics letters}} \textbf{75}, 2912--2914 (1999).

\bibitem{langer2012x}
M.~Langer, A.~Pacureanu, H.~Suhonen, \emph{et~al.}, \enquote{{X-ray phase nanotomography resolves the 3D human bone ultrastructure},} {\protect\JournalTitle{{PLoS ONE}}}  (2012).

\bibitem{kuan2020dense}
A.~T. Kuan, \enquote{Dense neuronal reconstruction through x-ray holographic nano-tomography,} {\protect\JournalTitle{Biophysical Journal}} \textbf{118}, 290a (2020).

\bibitem{andersson2020axon}
M.~Andersson, H.~M. Kjer, J.~Rafael-Patino, \emph{et~al.}, \enquote{Axon morphology is modulated by the local environment and impacts the noninvasive investigation of its structure--function relationship,} {\protect\JournalTitle{Proceedings of the National Academy of Sciences}} \textbf{117}, 33649--33659 (2020).

\bibitem{robisch2020nanoscale}
A.-L. Robisch, M.~Eckermann, M.~T{\"o}pperwien, \emph{et~al.}, \enquote{Nanoscale x-ray holotomography of human brain tissue with phase retrieval based on multienergy recordings,} {\protect\JournalTitle{Journal of Medical Imaging}} \textbf{7}, 013501--013501 (2020).

\bibitem{nguyen20213d}
T.-T. Nguyen, J.~Villanova, Z.~Su, \emph{et~al.}, \enquote{3d quantification of microstructural properties of lini0. 5mn0. 3co0. 2o2 high-energy density electrodes by x-ray holographic nano-tomography,} {\protect\JournalTitle{Advanced Energy Materials}} \textbf{11}, 2003529 (2021).

\bibitem{li2022dynamics}
J.~Li, N.~Sharma, Z.~Jiang, \emph{et~al.}, \enquote{Dynamics of particle network in composite battery cathodes,} {\protect\JournalTitle{Science}} \textbf{376}, 517--521 (2022).

\bibitem{martinez2016id16b}
G.~Mart{\'\i}nez-Criado, J.~Villanova, R.~Tucoulou, \emph{et~al.}, \enquote{Id16b: a hard x-ray nanoprobe beamline at the esrf for nano-analysis,} {\protect\JournalTitle{Journal of synchrotron radiation}} \textbf{23}, 344--352 (2016).

\bibitem{da2017high}
J.~C. da~Silva, A.~Pacureanu, Y.~Yang, \emph{et~al.}, \enquote{High-energy cryo x-ray nano-imaging at the id16a beamline of esrf,} in \emph{X-Ray Nanoimaging: Instruments and Methods III,}  vol. 10389 (SPIE, 2017), pp. 8--14.

\bibitem{kalbfleisch2022x}
S.~Kalbfleisch, Y.~Zhang, M.~Kahnt, \emph{et~al.}, \enquote{X-ray in-line holography and holotomography at the nanomax beamline,} {\protect\JournalTitle{Journal of Synchrotron Radiation}} \textbf{29}, 224--229 (2022).

\bibitem{hagemann2017probe}
J.~Hagemann, A.-L. Robisch, M.~Osterhoff, and T.~Salditt, \enquote{Probe reconstruction for holographic x-ray imaging,} {\protect\JournalTitle{Journal of synchrotron radiation}} \textbf{24}, 498--505 (2017).

\bibitem{bean2021new}
S.~Bean, V.~De~Andrade, A.~Deriy, \emph{et~al.}, \enquote{A new ultra-stable variable projection microscope for the aps upgrade of 32-id,} in \emph{11th Mechanical Engineering Design of Synchrotron Radiation Equipment and Instrumentation (MEDSI'20), Chicago, IL, USA, 24-29 July 2021,}  (JACOW Publishing, Geneva, Switzerland, 2021), pp. 211--214.

\bibitem{gureyev1996phase}
T.~E. Gureyev and K.~A. Nugent, \enquote{Phase retrieval with the transport-of-intensity equation. ii. orthogonal series solution for nonuniform illumination,} {\protect\JournalTitle{JOSA A}} \textbf{13}, 1670--1682 (1996).

\bibitem{barty1998quantitative}
A.~Barty, K.~Nugent, D.~Paganin, and A.~Roberts, \enquote{Quantitative optical phase microscopy,} {\protect\JournalTitle{Optics Letters}} \textbf{23}, 817--819 (1998).

\bibitem{zabler2005optimization}
S.~Zabler, P.~Cloetens, J.-P. Guigay, \emph{et~al.}, \enquote{Optimization of phase contrast imaging using hard x rays,} {\protect\JournalTitle{Review of Scientific Instruments}} \textbf{76} (2005).

\bibitem{paganin2006coherent}
D.~Paganin, \emph{Coherent X-ray optics}, 6 (Oxford University Press, USA, 2006).

\bibitem{langer2008quantitative}
M.~Langer, P.~Cloetens, J.-P. Guigay, and F.~Peyrin, \enquote{Quantitative comparison of direct phase retrieval algorithms in in-line phase tomography,} {\protect\JournalTitle{Medical physics}} \textbf{35}, 4556--4566 (2008).

\bibitem{salditt2008high}
T.~Salditt, S.~Kr{\"u}ger, C.~Fuhse, and C.~B{\"a}htz, \enquote{High-transmission planar x-ray waveguides,} {\protect\JournalTitle{Physical review letters}} \textbf{100}, 184801 (2008).

\bibitem{salditt2015compound}
T.~Salditt, M.~Osterhoff, M.~Krenkel, \emph{et~al.}, \enquote{Compound focusing mirror and x-ray waveguide optics for coherent imaging and nano-diffraction,} {\protect\JournalTitle{Journal of synchrotron radiation}} \textbf{22}, 867--878 (2015).

\bibitem{kruger2012sub}
S.~Kr{\"u}ger, H.~Neubauer, M.~Bartels, \emph{et~al.}, \enquote{Sub-10 nm beam confinement by x-ray waveguides: design, fabrication and characterization of optical properties,} {\protect\JournalTitle{Journal of synchrotron radiation}} \textbf{19}, 227--236 (2012).

\bibitem{hagemann2014reconstruction}
J.~Hagemann, A.-L. Robisch, D.~Luke, \emph{et~al.}, \enquote{Reconstruction of wave front and object for inline holography from a set of detection planes,} {\protect\JournalTitle{Optics Express}} \textbf{22}, 11552--11569 (2014).

\bibitem{maiden2009improved}
A.~M. Maiden and J.~M. Rodenburg, \enquote{An improved ptychographical phase retrieval algorithm for diffractive imaging,} {\protect\JournalTitle{Ultramicroscopy}} \textbf{109}, 1256--1262 (2009).

\bibitem{luke2004relaxed}
D.~R. Luke, \enquote{Relaxed averaged alternating reflections for diffraction imaging,} {\protect\JournalTitle{Inverse problems}} \textbf{21}, 37 (2004).

\bibitem{enders2016computational}
B.~Enders and P.~Thibault, \enquote{A computational framework for ptychographic reconstructions,} {\protect\JournalTitle{Proceedings of the Royal Society A: Mathematical, Physical and Engineering Sciences}} \textbf{472}, 20160640 (2016).

\bibitem{dora2024artifact}
J.~Dora, M.~M{\"o}ddel, S.~Flenner, \emph{et~al.}, \enquote{Artifact-suppressing reconstruction of strongly interacting objects in x-ray near-field holography without a spatial support constraint,} {\protect\JournalTitle{Optics express}} \textbf{32}, 10801--10828 (2024).

\bibitem{hagemann2017x}
J.~Hagemann, \emph{X-ray near-field holography: Beyond idealized assumptions of the probe} (G{\"o}ttingen University Press, 2017).

\bibitem{homann2015validity}
C.~Homann, T.~Hohage, J.~Hagemann, \emph{et~al.}, \enquote{Validity of the empty-beam correction in near-field imaging,} {\protect\JournalTitle{Physical Review A}} \textbf{91}, 013821 (2015).

\bibitem{stockmar2013near}
M.~Stockmar, P.~Cloetens, I.~Zanette, \emph{et~al.}, \enquote{Near-field ptychography: phase retrieval for inline holography using a structured illumination,} {\protect\JournalTitle{Scientific reports}} \textbf{3}, 1927 (2013).

\bibitem{Dai:00}
Y.~Dai, J.~Han, G.~Liu, \emph{et~al.}, \enquote{Convergence properties of nonlinear conjugate gradient methods,} {\protect\JournalTitle{SIAM Journal on Optimization}} \textbf{10}, 345--358 (2000).

\bibitem{Fletcher:64}
R.~Fletcher and C.~M. Reeves, \enquote{Function minimization by conjugate gradients,} {\protect\JournalTitle{The Computer Journal}} \textbf{7}, 149--154 (1964).

\bibitem{Polak:69}
E.~Polak and G.~Ribiere, \enquote{Note sur la convergence de m{\'e}thodes de directions conjugu{\'e}es,} {\protect\JournalTitle{ESAIM: Mathematical Modelling and Numerical Analysis-Mod{\'e}lisation Math{\'e}matique et Analyse Num{\'e}rique}} \textbf{3}, 35--43 (1969).

\bibitem{Polyak:69}
B.~T. Polyak, \enquote{The conjugate gradient method in extremal problems,} {\protect\JournalTitle{USSR Computational Mathematics and Mathematical Physics}} \textbf{9}, 94--112 (1969).

\bibitem{DaiYuan:99}
Y.~H. Dai and Y.~Yuan, \enquote{A nonlinear conjugate gradient method with a strong global convergence property,} {\protect\JournalTitle{SIAM J. Optim.}} \textbf{10}, 177--182 (1999).

\bibitem{goldstein1988satellite}
R.~M. Goldstein, H.~A. Zebker, and C.~L. Werner, \enquote{Satellite radar interferometry: Two-dimensional phase unwrapping,} {\protect\JournalTitle{Radio science}} \textbf{23}, 713--720 (1988).

\bibitem{ghiglia1998two}
D.~C. Ghiglia, \enquote{Two-dimentional phase unwrapping: Theory,} {\protect\JournalTitle{Algorithms, and Software}}  (1998).

\bibitem{guizar2011phase}
M.~Guizar-Sicairos, A.~Diaz, M.~Holler, \emph{et~al.}, \enquote{Phase tomography from x-ray coherent diffractive imaging projections,} {\protect\JournalTitle{Optics express}} \textbf{19}, 21345--21357 (2011).

\bibitem{wittwer2022phase}
F.~Wittwer, J.~Hagemann, D.~Br{\"u}ckner, \emph{et~al.}, \enquote{Phase retrieval framework for direct reconstruction of the projected refractive index applied to ptychography and holography,} {\protect\JournalTitle{Optica}} \textbf{9}, 295--302 (2022).

\bibitem{aidukas2024high}
T.~Aidukas, N.~W. Phillips, A.~Diaz, \emph{et~al.}, \enquote{High-performance 4-nm-resolution x-ray tomography using burst ptychography,} {\protect\JournalTitle{Nature}} \textbf{632}, 81--88 (2024).

\bibitem{becker2018atomic}
N.~G. Becker, A.~L. Butterworth, M.~Salome, \emph{et~al.}, \enquote{Atomic layer deposition of 2d and 3d standards for synchrotron-based quantitative composition and structure analysis methods,} {\protect\JournalTitle{Journal of Vacuum Science \& Technology A}} \textbf{36} (2018).

\bibitem{Aslan:19}
S.~Aslan, V.~Nikitin, D.~J. Ching, \emph{et~al.}, \enquote{Joint ptycho-tomography reconstruction through alternating direction method of multipliers,} {\protect\JournalTitle{Optics Express}} \textbf{27}, 9128--9143 (2019).

\bibitem{nikitin2019photon}
V.~Nikitin, S.~Aslan, Y.~Yao, \emph{et~al.}, \enquote{Photon-limited ptychography of 3d objects via bayesian reconstruction,} {\protect\JournalTitle{OSA Continuum}} \textbf{2}, 2948--2968 (2019).

\bibitem{nikitin2021distributed}
V.~Nikitin, V.~De~Andrade, A.~Slyamov, \emph{et~al.}, \enquote{Distributed optimization for nonrigid nano-tomography,} {\protect\JournalTitle{IEEE Transactions on Computational Imaging}} \textbf{7}, 272--287 (2021).

\end{thebibliography}

\end{document}